\documentclass{jpp}

\usepackage{amsmath, amsfonts,amssymb}
\usepackage{cancel}
\usepackage{xcolor}

\usepackage{graphicx}

\usepackage{epstopdf, epsfig}

\usepackage{bm}
\usepackage{hyperref}

\usepackage{natbib}
\usepackage{appendix}
\usepackage{enumitem} 

\newtheorem{definition}{Definition}

\newtheorem{theorem}{Theorem}
\newtheorem{lemma}[theorem]{Lemma}

\newcommand{\w}{\wedge }
\newcommand{\tp}{\tilde{p}}
\newcommand{\tA}{\tilde{A}}
\newcommand{\partiald}[2]{\frac{\partial#1}{\partial#2}}

\newcommand{\iotab}{{\iota\hspace{-0.6em}\raisebox{0.5pt}{\,--}}}

\newcommand{\mcal}[1]{\mathcal{#1}}

\newcommand{\bT}{\mathbb{T}}
\newcommand{\bR}{\mathbb{R}}
\newcommand{\cB}{\mcal{B}}
\newcommand{\cF}{\mcal{F}}
\newcommand{\cG}{\mcal{G}}
\newcommand{\cI}{\mcal{I}}
\newcommand{\cQ}{\mcal{Q}}
\newcommand{\cA}{\mcal{A}}

\newcommand{\cV}{\mcal{V}}

\newcommand{\dmat}{\mathbb{D}}

\newcommand{\depvar}{\bm{w}}

\newcommand{\Eq}[1]{(\ref{eq:#1})}
\newcommand{\Th}[1]{Th.~\ref{thm:#1}}
\newcommand{\Lem}[1]{Lem.~\ref{lem:#1}}

\newcommand{\Def}[1]{Def.~\ref{def:#1}}
\newcommand{\Sec}[1]{\S \ref{sec:#1}}

\newcommand{\App}[1]{App.~\ref{app:#1}}

\newcommand{\Step}[1]{Step~\ref{step:#1}}

\shorttitle{All field strengths are possible locally in Boozer coordinates}
\shortauthor{J.W. Burby, N. Duignan, T.J. Klotz, and J.D. Meiss}

\title{All field strengths are possible locally in Boozer coordinates}

\author{J.W. Burby\aff{1}, N. Duignan \aff{2}, T.J. Klotz \aff{3}\corresp{\email{taylor.klotz@colorado.edu}} and J.D. Meiss\aff{3}}

\affiliation{\aff{1}Department of Physics and Institute for Fusion Studies, The University of Texas at Austin, Austin, TX 78712, USA,
\aff{2}School of Mathematics and Statistics, University of Sydney, Sydney, Australia
\aff{3}Department of Applied Mathematics, University of Colorado, Boulder, CO, 80309-0526, USA}

\begin{document}

\maketitle

\begin{abstract}
The magnetic field strength is a central design object in the theory of magnetic confinement of plasma. Desirable confinement properties such as quasi-symmetry and omnigenity are characterized by special forms of the field strength when it is given in the particular coordinate system called ``Boozer coordinates." We ask a basic realizability question: which functions of Boozer coordinates can arise as the magnitude of a magnetic field that is written in Boozer coordinates? Using exterior differential systems and the Cartan-K\"ahler theorem, we show that every positive analytic function is locally realizable. Thus there is no local analytic obstruction to prescribing the field strength in Boozer coordinates; though further restrictions might arise from global geometric or additional physical constraints.
\end{abstract}

 
\section{Introduction: Definition of realizability}\label{sec:Intro}

\subsection{Motivation}

The magnetic-field strength $|\bm{B}|$ is a central design object in modern stellarator theory and optimization. The confinement of guiding-center trajectories \cite{Littlejohn_1983,Littlejohn_1981,Northrop_1963} is strongly governed by the variation of $|\bm{B}|$ on magnetic surfaces, and several of the most important classes of magnetic fields with favorable confinement properties are naturally characterized by enforcing special functional forms of $|\bm{B}|$ \cite{Helander14,Imbert25}.
These properties can be most easily stated by specifying $|\bm{B}|$ as a function, $\cB(\psi,\theta,\zeta)$, in terms of  \textit{Boozer coordinates} \cite{Boozer_pf_1981} $(\psi,\theta,\zeta)$ (defined more concretely below in \Def{Boozer}).

In particular, quasi-symmetry requires $\cB$ to depend on a single linear combination of the Boozer angles $(\theta,\zeta)$ \cite{boozerTransportIsomorphicEquilibria1983}. Similarly, quasi-isodynamic \cite{goriQuasiisodynamicStellarators1997}, isodrastic \cite{Burby23b}, omnigenous \cite{caryOmnigenityQuasihelicityHelical1997}, and isoprominent \cite{Burby23} configurations impose geometric restrictions on the contours and critical sets of $\cB$. Consequently, a substantial part of stellarator design can be viewed as the problem of engineering a desirable magnetic-field strength in Boozer coordinates.  This viewpoint underlies both near-axis constructions \cite{senguptaStellaratorEquilibriumAxisexpansion2024,jorgeNearaxisExpansionStellarator2020,landremanMappingSpaceQuasisymmetric2022,garrenExistenceQuasihelicallySymmetric1991} that relate a prescribed field strength to the geometry of magnetic surfaces \cite{landremanDirectConstructionOptimized2018} and numerical stellarator optimization.  For example, highly quasi-symmetric configurations have been obtained by minimizing symmetry-breaking Fourier components of $|\bm{B}|$
in Boozer coordinates \cite{landremanMagneticFieldsPrecise2022}, while more recent optimization schemes explicitly target functions $\cB$ with quasi-isodynamic or omnigenous properties, and explicitly minimize the difference between a target field strength and that of the computed equilibrium \cite{goodmanConstructingPreciselyQuasiisodynamic2023,dudtMagneticFieldsGeneral2024}.

These approaches pose a basic inverse problem that precedes questions of optimization, force balance, stability, or coil design: 
\begin{center}
    \emph{Which functions $\cB(\psi,\theta,\zeta)$ can actually occur as the magnitude of a magnetic field in Boozer coordinates?}
\end{center} 
The relationship between a prescribed $\cB$ in Boozer coordinates and the corresponding magnetic field and its geometry in $\bR^3$ is highly nonlinear, even for near-axis expansions. Moreover, numerical studies often note that the realizability of a prescribed field-strength is a fundamental issue \cite{landremanDirectConstructionOptimized2018,dudtMagneticFieldsGeneral2024}.

In this paper, we isolate the local, geometric part of the question concerned with realizing $\cB$ in Boozer coordinates, leaving open the further questions of global realizability and satisfying force-balance conditions.  Our main result, \Th{FTSD}, shows that---perhaps surprisingly---there is no local analytic obstruction: every positive analytic function $\cB(\psi,\theta,\zeta)$ is locally realizable as the magnitude of a magnetic field admitting Boozer coordinates. Moreover, we quantify how large the space of such realizing functions $\cB$ is: it depends on two functions of one variable and three functions of two variables.

\subsection{Precise formulation of the realizability problem}

To formulate the realizability problem precisely, we first recall the geometric setting in which Boozer coordinates arise. The treatment of magnetically confined plasmas is facilitated by the use of \emph{flux coordinates}, adapted to a foliation by invariant toroidal magnetic surfaces \cite{Kruskal58,Helander14}. Let $\psi$ label these surfaces and let $(\theta,\zeta)\in\bT^2$ denote poloidal and toroidal angles on each surface.  A magnetic field tangent to the flux surfaces can then be written in the form
\begin{equation}\label{eq:FluxCoord}
    \bm{B} = \nabla \chi \times \nabla (\theta - \iotab\zeta) \;.
\end{equation}
Here $2\pi \chi (\psi)$ is the toroidal magnetic flux and $\iotab(\psi)$ is the rotational transform. Of course, in this case $B \cdot \nabla{\psi} = 0$: the field is tangent to flux surfaces. Special cases of these coordinates include Hamada coordinates, where the $(\psi,\theta,\zeta)$ coordinate system is chosen so that its Jacobian,
\begin{equation}\label{eq:Jacobian}
    \rho = \nabla \psi \cdot \nabla \theta \times \nabla \zeta \;,
\end{equation} 
is a flux function, and symmetry coordinates where the general toroidal angle $\zeta$ is taken to be the toroidal symmetry coordinate. Flux coordinates form a foundation for studying the confinement of particles as well as the development of instabilities and flows in a confined plasma \cite{Hazeltine03, Imbert25}.

If a magnetic field satisfies the magnetohydrostatic (MHS) equation
\begin{equation}\label{eq:MHS}
    \bm{J} \times \bm{B} = \nabla p \;, 
\end{equation} 
for \textit{some} divergence free vector field $\bm{J}$ and a pressure $p(\psi)$, with $\nabla p \neq 0$ almost everywhere, we say it is \textit{MHS integrable}. Indeed, \Eq{MHS} implies that $[\bm{B},\bm{J}]=0$, $\bm{B} \cdot \nabla p = 0$ and $\bm{J} \cdot \nabla p = 0$, so $p$ is an integral of both field line flows and $\bm{B},\bm{J}$ commute. When $\bm{J} = \nabla \times \bm{B}$ is the current, Hamada showed MHS integrability implies that, in the neighborhood of any flux surface, there exist coordinates in which the contravariant components
\[
    B^{\theta}(\psi) =B \cdot \nabla \theta, \quad B^\zeta(\psi) = B \cdot \nabla \zeta \;,
\]
are flux functions \cite{Hamada_1962}. For Hamada coordinates the flows of both $\bm{B}$ and $\bm{J}$ are straight on each flux surface. Moreover, the ratio $\iotab = B^\theta/B^\zeta$ gves the rotational transform as a flux function. Such coordinates also exist more generally when $\bm{J}$ is any divergence free vector field \cite{Burby21}.

For the case of an MHS-integrable field, \cite{Boozer_pf_1981}  introduced an alternative coordinate system by requiring that the angular covariant components of $\bm{B}$ are flux functions, i.e., that
\begin{equation}\label{eq:BoozerCovar}
    \bm{B} = B_\theta(\psi) \nabla \theta + B_\zeta(\psi) \nabla \zeta + B_\psi(\psi,\theta,\zeta) \nabla \psi \;.
\end{equation}
In this case the vector fields $\bm{B}$ and $\bm{B} \times \nabla p = |\bm{B}|^2 J_\perp$ have straight field lines on each flux surface. More importantly for the present work, Boozer coordinates provide the natural coordinate system in which confinement properties such as quasi-symmetry and omnigenity are expressed directly as conditions on $|\bm{B}|$. For example, the notion of quasi-symmetry (QS) is when $|\bm{B}|=\cB(\psi,M\,\theta-N\,\zeta)$ for integers $M,N$ \cite{Helander14,Imbert25}.

Combining \Eq{BoozerCovar} with the flux coordinates \Eq{FluxCoord} gives:

\begin{definition}[Boozer Coordinates]\label{def:Boozer}
A vector field $\bm{B}$ admits \textbf{Boozer coordinates} $\psi \in (\psi_l,\psi_u)$ and $(\theta,\zeta) \in \bT^2$ for a flux interval $(\psi_l, \psi_u)$, if there exist smooth flux functions $F,G,K,L :(\psi_l, \psi_u) \to \bR$ such that 
\begin{align}
    \bm{B} &= F\nabla \psi \times \nabla \theta - G \nabla \psi \times \nabla \zeta \;, \label{eq:BoozerC1}\\
    \bm{B} \times \nabla \psi & = K \nabla \psi \times \nabla \theta - L\nabla \psi \times \nabla \zeta \;.\label{eq:BoozerC2}
\end{align}
\end{definition}

There are several immediate consequences of this definition. Equation \Eq{BoozerC1} implies that $B^\theta = G \rho$ and $B^\zeta = F \rho$, for the Jacobian \Eq{Jacobian}. Since $\rho$ is not assumed to be a flux function, this contrasts with the case of Hamada coordinates. However \Eq{BoozerC2} implies that $B_\theta = -K$ and $B_\zeta = L$, the flux functions in the covariant form \Eq{BoozerCovar}

In \Def{Boozer} we did not assume that $\nabla \cdot \bm{B}$ vanishes; nevertheless, any vector field that satisfies \Eq{BoozerC1}-\Eq{BoozerC2} is necessarily divergence free.  Indeed from \Eq{BoozerC1}
\[
    \nabla \cdot \bm{B} = \nabla F \cdot \nabla \psi \times \nabla \theta- \nabla G \cdot \nabla \psi \times \nabla \zeta = 0
\]
since $F$ and $G$ are flux functions.  Moreover, whenever $\nabla \psi \neq 0$,  \Eq{BoozerC1} implies that $\bm{B}$ is tangent to the flux surface. So is the current $\bm{J} = \nabla \times B$ since 
\[
    J^\psi = \nabla \psi \cdot J  =  \nabla \psi \cdot (\nabla \times \bm{B}) = \nabla \cdot(\bm{B} \times \nabla \psi) \;,
\]
and \Eq{BoozerC2} gives
\[
    \nabla \cdot(\bm{B} \times \nabla \psi) = \nabla K \cdot \nabla \psi \times \nabla \theta- \nabla L \cdot \nabla \psi \times \nabla \zeta = 0 \;.
\]
Therefore $J^\psi = 0$.

Another consequence of \Def{Boozer} is
\begin{lemma}\label{lem:PositiveJacobian}
    Assume $(\psi,\theta,\phi)$ are Boozer coordinates for a nonzero field $\bm{B}$. Then 
    \begin{equation}\label{eq:positivity}
        |\bm{B}|^2 = \rho \Delta, \quad \Delta \equiv FL-GK\;.
    \end{equation}
    where $\rho$ is the Jacobian \Eq{Jacobian}. In particular, $\Delta > 0$.
\end{lemma}

\begin{proof}   
    On one hand, since $\bm{B} \cdot \nabla \psi = 0$, the identity 
    $\bm{a}\times (\bm{b} \times \bm{c}) = (\bm{a}\cdot\bm{c}) \bm{b} - (\bm{a}\cdot \bm{b}) \bm{c}$ gives
    \[
        \bm{B}\times (\nabla \psi \times \bm{B}) = |\bm{B}|^2 \nabla \psi \;.
    \]
    On the other hand, the identity 
    $(\bm{a} \times \bm{b})\times (\bm{a}\times\bm{c}) = (\bm{a}\cdot (\bm{b}\times\bm{c})) \bm{a}$ and \Eq{BoozerC1}-\Eq{BoozerC2} give
    \begin{align*}
        \bm{B}\times (\nabla \psi \times \bm{B}) &=  \Delta\left(\nabla \psi \times \nabla \theta \right) \times \left(\nabla \psi \times \nabla \zeta \right) \\
         &= \rho \Delta \nabla \psi  \;.
    \end{align*}
    Since $\psi$ is a coordinate function, it must be that $\nabla \psi \neq 0$ and therefore equating these two results  gives \Eq{positivity}. Finally, for $(\psi,\theta,\phi)$ to be a well-defined, right-handed coordinate system, then the determinant of the Jacobian, $\rho$ must be positive. It follows that $\Delta = FL-GK > 0$. 
\end{proof}

We are now in a position to give a definition of what we will call the \textit{Boozer realizability problem}. The idea is that we would like to specify the magnitude $|\bm{B}|$ in Boozer coordinates and ask the question: When is there a magnetic field that admits Boozer coordinates in which $|\bm{B}|$ has the specified form?
We will let $\bm{x}$ be Euclidean coordinates on $\bR^3$.
We suppose that a magnitude function $\cB(\psi, \theta,\zeta) > 0$ is given. 
Realizing $\cB$ in Boozer coordinates is equivalent to finding functions $\psi(\bm{x}),\theta(\bm{x}),\zeta(\bm{x})$ and $F,G,K$, and $L$ such that 
\begin{equation}\label{eq:BoozerPDE}
\begin{aligned}
    |\bm{B}(\bm{x})| &= \cB(\psi(\bm{x}), \theta(\bm{x}), \zeta(\bm{x}) ) \;, \\
     \bm{B}(\bm{x}) &= F(\psi(\bm{x}))\nabla \psi(\bm{x}) \times \nabla \theta(\bm{x}) - G(\psi(\bm{x})) \nabla \psi(\bm{x}) \times \nabla \zeta(\bm{x}) \;, \\
     \bm{B} \times \nabla \psi & = K(\psi(\bm{x})) \nabla \psi(\bm{x}) \times \nabla \theta(\bm{x}) - L(\psi(\bm{x}))\nabla \psi(\bm{x}) \times \nabla \zeta(\bm{x}) \;. 
\end{aligned}
\end{equation}

Note that by \Lem{PositiveJacobian} we will then have $\Delta = FL-GK >0$.

\begin{definition}[Local Boozer realizability]\label{def:Realizability}
Let $(\psi,\theta,\zeta) \in \cA$, an open subset of $\bR^3$. The magnitude ${\cB}:\cA \rightarrow \bR^+$ is \textbf{locally realizable} in Boozer coordinates if there is an open set $Q\subset\bR^3$ and a vector field $\bm{B}:Q\rightarrow\bR^3$ such that
there exists a coordinate change $\Phi: Q \to \cA$ such that
\[
    \Phi(\bm{x}) = (\psi(\bm{x}), \theta(\bm{x}), \zeta(\bm{x}))
\]
satisfies \Eq{BoozerPDE} for some flux functions $F,G,K$, and $L$.
\end{definition}

Note that we specifically allow \textit{any} open set $\cA$ in \Def{Realizability}. These coordinates would be \textbf{global} if they exist for a toroidal annulus, i.e.,  $\cA \to (\psi_l, \psi_u) \times \bT^2$, so that $Q$ becomes a toroidal annular region in $\bR^3$. 

Thus Boozer realizability can be stated as a question about the existence of a solution to the system \Eq{BoozerPDE}, where the unknowns are $\psi,\bm{B},F,G,K,L$. Note that $\cB$ must be positive: zeros of the magnetic field may preclude existence. A priori, the required existence of a solution to \Eq{BoozerPDE} may constrain the flux functions $F,G,K$, and $L$. However, our main result \Th{FTSD} establishes that \textit{any} choice of these functions suffices for local realizability, provided $\Delta > 0$.

A similar, and comparatively simpler, version of this question can be found in \cite[Section 4.2]{Burby23b}, where it is proven that a local version of what one might call \textit{Clebsch realizability} holds. The proof invokes the Cauchy-Kowalevski theorem (that we recall in \Sec{CKAlg}). 

We will show in \Sec{BoozerEDS}---\ref{sec:Proof} below, that any nonzero magnitude $\cB$ can be realized, at least locally:

\begin{theorem}\label{thm:FTSD}
Every positive, analytic function $\cB:\cA \to\bR^+$ and set of given analytic coefficients $F,G,K,L: (\psi_l, \psi_u) \to \bR$ that satisfy $\Delta >0$ is locally realizable in Boozer coordinates. Moreover, the system of Boozer coordinates $(\psi,\theta,\zeta) = \Phi(\bm{x})$ for a desired magnitude $\cB$ may be determined by a sequence of Cauchy problems with initial data depending on two functions of one variable and three functions of two variables.
\end{theorem}

The proof relies on real analyticity upon invoking the Cauchy-Kowalevski theorem and its generalization, the Cartan-K\"ahler algorithm for \textit{exterior differential systems} (EDS). EDS have previously been used to study the local realizability of force-free magnetic fields with prescribed proportionality factor \cite{clellandBeltramiFieldsNonconstant2020}, and Cartan-K\"ahler methods have recently been proposed for the local existence problem associated with admissible quasi-symmetries \cite{burbyCharacterizationAdmissibleQuasisymmetries2025}.
We give a brief overview of this algorithm, together with simple examples in \Sec{CKAlg}---\ref{sec:Examples}. It is important to note that the Cartan-K\"ahler theorem implies local existence, but the theory does not necessarily preclude extensions of such solutions globally. 

Our proof of \Th{FTSD} uses differential form notation (for a tutorial see \cite{MacKay20a}). To reformulate the PDE system \Eq{BoozerPDE} in terms of forms,
recall that the covariant representation for coordinates $(\xi^1,\xi^2,\xi^3)$ is $\bm{B}^\flat = B_i d\xi^i$, and the contravariant representation can be expressed as $\star \bm{B}^\flat = \iota_{\bm{B}} \Omega$ where $\Omega = \rho d\xi^1 \w d\xi^2 \w d\xi^3$ is the volume form and $\star$ is the Hodge star operator. Boozer coordinates then correspond to the existence of a coordinate system for which
\begin{equation}\label{eq:BoozerDiffForm}
\begin{aligned}
    \Phi_*(|\bm{B}|) &= \cB \;, \\
    \Phi_*(\iota_{\bm{B}}\Omega) &= F(\psi)\,d\psi\w d\theta - G(\psi)\,d\psi\w  d\zeta \;,  \\
    \Phi_*(\bm{B}^\flat \w  d\psi) &= K(\psi)\,d\psi\w  d\theta - L(\psi)\,d\psi\w  d\zeta \;.
\end{aligned}
\end{equation}
Here  $\Phi_*$ is the action of the transformation on forms, i.e., the push-forward operator.
This set of equations can be thought of as an EDS, see \Lem{EDSReal}. The proof of \Th{FTSD} in  \Sec{BoozerEDS}-\ref{sec:Proof} will use this notation.

\section{Methodological background}\label{sec:CKAlg}

In this section we give an overview of the concepts and terminology necessary to prove \Th{FTSD}. These ideas were developed by Cartan in the early 1900's when he conceived of a far-reaching geometric framework for showing local existence of analytic solutions to PDE systems. This would later be called the Cartan-K\"ahler theorem or ``algorithm". Cartan's approach is to study the more general class of objects called exterior differential systems, which when written in local coordinates are essentially PDEs with mild non-degeneracy conditions. The exposition of Cartan's work on this subject can be notoriously difficult to penetrate for non-experts, and so in this section we try to given an overview of the Cartan-K\"ahler algorithm using as little of the differential geometry and differential forms language as possible. Moreover, the methodology is demonstrated through a simple example. This same example is revisited in \App{CK-forms-ver} where a more compact version of the Cartan-K\"ahler algorithm is implemented in the language of differential forms. 

The Cartan-K\"ahler theorem is an \textit{if and only if} statement on the local existence of analytic solutions to a PDE system and it is powered inductively by the Cauchy-Kowalevski existence theorem for analytic systems.

\subsection{Cauchy-Kowalevski Theorem and Sequences of Cauchy Problems}\label{sec:CK-CauchySeq}

The famous Cauchy-Kowalevski theorem gives requirements on a system of first order PDEs ensuring existence of a local solution given by Taylor series. 

\begin{theorem}[Cauchy-Kowalevski \cite{Courant89, Olver95, EDS-Book, Ivey16}]\label{thm:C-Kowalevski}
Consider a system of $d$ first order PDEs with $d$ dependent variables $\bm{u} = (u^1,u^2,\ldots, u^d)$ and $n+1$ independent variables $(\bm{x},t) \in \bR^{n+1}$ of the form
\begin{equation}\label{eq:CK-Problem}
    \partiald{u^a}{t}=\cF^a\left(\bm{x},t,\bm{u},D\bm{u}\right)
            =\sum_{b=1}^d \sum_{i=1}^n A^a_{bi}(\bm{x},t,\bm{u})\partiald{u^b}{x^i}+\tA^a(\bm{x},t,\bm{u}), \quad a = 1, \ldots, d \;, 
\end{equation}
where the functions $A^a_{bi}$ and $\tA^a$ are real-analytic on an open set $Q \subset \bR^{n+1} \times \bR^d$.
Assume that  at $t=t_0$, the initial values, 
\begin{equation}\label{eq:CK-InitialData}
    u^a (\bm{x},t_0) = g^a(\bm{x}) \;, \quad a = 1, \ldots, d \;,
\end{equation}
are real-analytic on an open set $V\subset \bR^n$ and that for each $\bm{x}\in V$ we have $(\bm{x},t_0,\bm{g}(\bm{x}))\in Q$. Then for any point $\bm{x}_0 \in V$, there is  a solution $u^a(\bm{x},t)$ of \Eq{CK-Problem} that is real-analytic on an open set of $\bR^{n+1}$ containing $(\bm{x}_0,t_0)$.
\end{theorem} 

The system \Eq{CK-Problem} can be thought of as an evolution equation in the time-like variable $t$ with initial data \Eq{CK-InitialData}. A  PDE that can be written in this form is said to be in Cauchy-Kowalevski form, or is simply called a \textit{Cauchy problem}. 

There is no guarantee that any given PDE system can be written in the form \Eq{CK-Problem}, even allowing for changes of dependent and independent variables. Nevertheless, even when Cauchy-Kowalevski does not directly apply (e.g., for overdetermined PDE systems), it might still be possible to interpret a first order PDE system as giving an ``evolution" along a sequence of variables, none of which may have a physical time interpretation. In this case it may still be possible to apply a generalization of \Th{C-Kowalevski}, or conclude that no analytic solutions exist. 

As an example, consider the nonlinear PDE system of $d=4$ functions, $\bm{u}=(u^1,u^2,u^3,u^4)$, of $n=3$ variables, $\bm{x}=(x^1,x^2,x^3)$:
\begin{equation}\label{eq:CK-Example}
\begin{split}
    \partiald{u^1}{x^2}&=u^2-u^1+u^3u^4 \;,\\
    \partiald{u^2}{x^2}&=u^1-u^2 \;,
\end{split}
\quad
\begin{split}
    \partiald{u^1}{x^3}&=1 \;,\\
    \partiald{u^2}{x^3}&=1 \;,\\
    \partiald{u^3}{x^3}&=u^3 \;,\\
    \partiald{u^4}{x^3}&=-u^4 \;.
\end{split}
\end{equation}
The system \Eq{CK-Example} is an example of the more general first order system
\begin{equation}\label{eq:CK-Sequence}
    \partiald{\bm{u}^a}{x^i} = \cF^{a}_i( \bm{x},\bm{u}), \quad 
            a = 1,\ldots, d, \; i = 1, \ldots, n \;.
\end{equation}
However, notice that for \Eq{CK-Example}, there are no evolution equations in the $x^1$ direction. Consequently, it is not clear from the outset what `initial' data (if any) can be specified to get a well-defined PDE problem.  In general, one might be able to reformulate the PDE system \Eq{CK-Sequence} as a sequence of Cauchy problems as in \Th{C-Kowalevski}, but even so such systems are often overdetermined. 

Any solution $\bm{u}(\bm{x})$ of \Eq{CK-Sequence} must satisfy a \textit{compatibility condition}; namely, when the solution is inserted into \Eq{CK-Sequence}, the mixed, second partial derivatives of the $u^a$ must commute; therefore,
\begin{equation}\label{eq:Compatible}
    \frac{d}{d x^i} \cF^{a}_j(\bm{x}, \bm{u}(\bm{x})) -
       \frac{d}{d x^j}\cF^{a}_i(\bm{x},\bm{u}(\bm{x}))=0 \;.
\end{equation}
These necessary compatibility conditions sometimes show up in an expression referred to as \textit{torsion}, as we discuss in \Sec{Cartan-Kahler}.

It can be shown by elementary methods (see \App{Example}) that there is suitable initial data such that the example \Eq{CK-Example} admits local analytic solutions. The solutions can be obtained by solving two Cauchy problems, the first for the $x^2$ evolution and then the second for the $x^3$ evolution, with the initial data for the PDE being the relevant Cauchy data. Compatibility can then be directly checked from these solutions, in particular one can check \Eq{Compatible} for $\cF^a_2$ and $\cF^a_3$ given by \Eq{CK-Example}. We will reformulate \Eq{CK-Example} as a system of differential forms in \Sec{Examples}, as a simple example of the Cartan-K\"ahler algorithm.

\subsection{Cartan-K\"ahler Algorithm}\label{sec:Cartan-Kahler}
Generally, it is not at all obvious when one can arrange a given PDE into a system of the form \Eq{CK-Sequence}, and it is even less obvious when the compatibility conditions \Eq{Compatible} will be satisfied. Luckily, there is a geometric ``algorithm" that can help determine whether this is possible. The key objects of study in the algorithm are the \textbf{tableau} and the \textbf{torsion} of a linear \textbf{Pfaffian system}.

Torsion detects explicit algebraic and first order compatibility conditions. The tableau is a subspace of rectangular matrices encoding the highest derivative terms in the PDE. Torsion must vanish for any solution to exist; if it does not then additional conditions to make it vanish must be appended to the original PDE system and the algorithm must restart. Assuming torsion vanishes, the pattern of the  independent entries and ``degrees of freedom" of the tableau essentially determine the existence of analytic solutions via a sequence of Cauchy problems. 

The Cartan-K\"ahler theorem assumes an analytic PDE system with (a) no torsion, and (b) a tableau that satisfies a linear algebraic condition called Cartan's test. It proves that one can extend any $j$-dimensional solution transversely to obtain a unique $(j+1)$-dimensional solution by an application of the Cauchy-Kowalevski theorem; here the $j$-dimensional solution plays the role of initial conditions. The Cartan-K\"ahler theorem is usually assumed to be applied inductively in $j$ until one has a solution of maximal dimension. The necessary conditions (a) and (b) then become sufficient conditions in the Cartan-K\"ahler theorem through the mechanism of the Cauchy-Kowalevski theorem.

We give here a compact overview of this algorithm; for more details see, for example, \cite{Olver95, Ivey16,EDS-Book}.

Given a PDE system, there is a four step process to determine existence of local analytic solutions:
\begin{enumerate}
\item\label{step:EDS} \textbf{Encode PDE as a Pfaffian System}, see \Sec{Pfaffian_Defined}, and the differential forms method described in \App{CK-forms-ver}.

\item \textbf{Find the associated Tableau}, see \Sec{Pfaffian_Tableau}.

\item \textbf{Check if Torsion vanishes}, see \Sec{Pfaffian_Torsion}. 

\item\label{step:Cartan} \textbf{Check if Tableau is involutive}, see \Sec{Linear_PDEs}. 

{\textbf{Warning:} On rare occasions, one gets a false negative for failure of equality for Cartan's test. This can be a subtle matter but is usually handled by changing the basis forms slightly. A more reliable method, replacing step (4), is given in  \cite{Olver95}.}
\end{enumerate}

The Cartan-Kuranishi theorem \cite{EDS-Book} guarantees that by adding a finite number of additional differential equations we can either express the original PDE plus the additional derivative conditions in Cauchy-Kowalevski form, or the PDE admits no non-trivial analytic solutions. This is a landmark result in exterior differential systems theory. 

In \Sec{Examples} we will use the example system \Eq{CK-Example} to demonstrate the steps in the Cartan-K\"ahler algorithm as listed above. 

\subsection{Linear PDEs: Tableau and Cartan's Test}\label{sec:Linear_PDEs}

Our analysis of realizability requires an understanding of both tableau and Cartan's test for involutivity of tableau. In particular, our analysis of the realizability theorem, \Th{FTSD},  which is proved in \Sec{Proof}, uses Cartan's test. Here we motivate and describe both tableau and Cartan's test using elementary linear algebra. Suppose that $\bm{u}: \bR^n \to \bR^{d}$ are ${d}$ functions of $\bm{x} \in \bR^n$ that solve
 a \textbf{simple} system of partial differential equations: that is,  ${c}$ scalar PDEs that are linear, first-order, and homogeneous:

\begin{equation}\label{eq:SimplePDE}
    \sum_{i=1}^n L^i \partial_i \bm{u} = 0 \;,
\end{equation}
for $n$ \textit{constant} rectangular matrices $L^i \in \bR^{{c} \times {d}}$. A \textbf{tableau}, $\Pi$, is the vector space of homogeneous \textit{linear} solutions of such a PDE system. In symbols,
\begin{equation}\label{eq:SimplePDETableau}
    \Pi = \left\{\bm{u}:\bR^n\rightarrow\bR^{d}: \exists A \in \bR^{{d}\times n} \text{ s.t. } \bm{u}(\bm{x}) = A\,\bm{x}, \text{ and } \sum_{i=1}^n L^i \partial_i \bm{u} = 0\right\} \;.
\end{equation}
Since every $\bm{u}$ in the tableau may be identified with its associated $A$-matrix, every tableau may be identified with a linear subspace of ${d}\times n$ matrices. 

Conversely, it is straightforward to show that \textit{any} vector subspace of ${d}\times n$ matrices can be realised as the tableau of some simple PDE system of the form \Eq{SimplePDE} with $n$ independent variables and ${d}$ dependent variables.  This implies we could define tableau more simply as merely linear subspaces of $\bR^{{d} \times n}$, though this would de-emphasize the connection between tableau and simple PDE systems. It is nevertheless convenient to switch freely between viewing tableau as homogeneous linear solutions of simple systems and as subspaces of rectangular matrices. In practice, tableau often appear first as subspaces of matrices; their interpretation in terms of simple systems is then merely a computational convenience.

Given a tableau \Eq{SimplePDETableau}, its \textbf{first prolongation}, $\Pi^{(1)}$,  is the vector space of all homogeneous \emph{quadratic} solutions of the simple system. Let $\cQ$ denote the space of symmetric bilinear maps $q:\bR^n\times\bR^n\rightarrow\bR^{d}$. Since homogeneous polynomials of degree-two are in one-to-one correspondence with $\cQ$, the first prolongation of $\Pi$ may be written
\[
\Pi^{(1)} = \left\{\bm{u}:\bR^n\rightarrow\bR^{d} : \exists q\in \cQ \text{ s.t. }\bm{u}(\bm{x}) = \tfrac{1}{2}q(\bm{x},\bm{x}), \text{ and }\sum_{i=1}^n L^i \partial_i \bm{u} = 0\right\}.
\]

The following counting argument leads to an upper bound on $\dim \,\Pi^{(1)}$. Let $e_i$, $i=1,\dots, n$, denote the standard basis for $\bR^n$. Given $\bm{u}\in \Pi^{(1)}$, notice that that for each $i = 1, \ldots, n$, the derivative $\partial_i \bm{u}(\bm{x}) = q(e_i,\bm{x}) = A_i \bm{x} $ is a homogeneous linear solution of \Eq{SimplePDE}. Therefore $A_1 \in \Pi$, and so has at most $\dim(\Pi)$ free elements. There are fewer \emph{additional} free elements in $\partial_2\bm{u}(\bm{x}) = A_2\,\bm{x}$ because equality of mixed partials implies that $A_2e_1 = \partial_1\partial_2 \bm{u} = \partial_2\partial_1\bm{u} = A_1e_2$. Thus the first column of $A_2$ is completely determined by the second column of $A_1$. This is an example of the compatibility condition \Eq{Compatible}. The collection of homogeneous linear solutions with specified $A\,e_1$ is an affine subspace of $ \Pi$ and so differences between such solutions lie in a linear subspace $\Pi_1\subset \Pi$ consisting of all solutions of \Eq{SimplePDE} with $\partial_1\bm{u} = 0$. The number of free derivatives among $\partial_1\bm{u}$ and $\partial_2 \bm{u}$ is therefore at most $\dim(\Pi) + \dim(\Pi_1)$. More generally, define
\[
    \Pi_i = \{ \bm{u}\in \Pi \;:\; \partial_1\bm{u} = \partial_2 \bm{u} = \dots = \partial_i\bm{u} = 0 \} \subset \Pi_{i-1} \;,\quad \Pi_0 = \Pi \;.
\]
Continuing the preceding argument inductively implies $\dim\,\Pi^{(1)}\leq \dim(\Pi) + \dim(\Pi_1) + \dots + \dim(\Pi_n)$. {Note that $\dim \Pi_n = 0$}. The differences between these dimensions, 
\[
     s_i = \dim(\Pi_{i-1}) - \dim(\Pi_{i}) \;,
\]
are called  the \textbf{Cartan characters}. Using these gives the upper bound
\begin{align}
    \dim\,\Pi^{(1)}\leq s_1 +2s_2 + 3s_3 + \dots +n s_n.\label{eq:cartan_test_inequality}
\end{align}
A tableau is \textbf{involutive} if this upper bound gives the true dimension of $\Pi^{(1)}$ when computed using a generic basis for both $\bR^n$ and $\bR^{d}$.

Checking for involutivity by computing the Cartan characters is known as \textbf{Cartan's test}. 

\subsection{Linear Pfaffian systems\label{sec:Pfaffian_Defined}}
In this section we define the concept of a linear Pfaffian system, and discuss the importance of it satisfying an independence condition. At first glance, focusing on linear Pfaffian systems appears to be focusing on a rather restrictive class of PDEs. However, as elucidated on in \App{PDE_to_Pfaffian}, most PDEs can be written as a linear Pfaffian system. 

To start, let $(\bm{x},\bm{y}) \in \bR^n\times\bR^\ell$. 
We think of $\bm{x}\in\bR^n$ as independent variables and $\bm{y}\in\bR^\ell$ as dependent variables for a system of quasilinear first-order partial differential equations. Such a system can be encoded as a  \textbf{Pfaffian system}---a collection of $r$ linearly-independent $1$-forms, $\alpha^a$, on $\mathbb{R}^n\times \mathbb{R}^\ell$. 

Each 1-form can be written in terms of its $\bm{x}$ and $\bm{y}$ components as 
\begin{equation}\label{eq:alphaxy}
    \alpha^a = (\alpha^a_{\bm{x}})^T d\bm{x} + (\alpha^a_{\bm{y}})^T d\bm{y} \;, 
\end{equation}
for vector functions $\alpha^a_{\bm{x}}(\bm{x},\bm{y})\in\bR^n$ and $\alpha^a_{\bm{y}}(\bm{x},\bm{y})\in\bR^\ell$. To define a PDE using such a system of $1$-forms, treat $\bm{y} \to \hat{\bm{y}}(\bm{x})$ as dependent variables  for functions $\hat{\bm{y}}: \bR^n \to \bR^\ell$. Then  the vanishing of the forms can be interpreted as a PDE by the replacement $d\bm{y} \to D\hat{\bm{y}}(\bm{x})\,d\bm{x}$ where $D\hat{\bm{y}}(\bm{x})$ denotes the $\ell\times n$ Jacobian matrix whose rows are the gradients of the components of $\hat{\bm{y}}$. The resulting PDE system is
\begin{align}\label{eq:general_pfaffian_pde}
    \alpha^a_{\bm{x}}(\bm{x},\hat{\bm{y}}(\bm{x}))^T + \alpha^a_{\bm{y}}(\bm{x},\hat{\bm{y}}(\bm{x}))^T D \hat{\bm{y}}(\bm{x}) = 0,\quad a=1,\dots, r.\
\end{align}
Note that this is a first-order system of $r$, vector PDEs with $n$ independent variables. These PDEs are quasilinear since the coefficients do not depend upon derivatives.  

\begin{definition}[Independence Condition]\label{defn:independence}
The Pfaffian system \Eq{alphaxy} satisfies an \textbf{independence condition} if the vectors $\alpha_{\bm{y}}^a(\bm{x},\bm{y})\in\bR^\ell$, $a=1,\dots, r$, are linearly independent for each $(\bm{x},\bm{y})\in \mathbb{R}^n\times \mathbb{R}^\ell$.
\end{definition}

Any analytic solution of \Eq{general_pfaffian_pde} near a point $\bm{x}_0 \in \bR^n$ has a Taylor series expansion of the form
\begin{align}\label{eq:Taylor_Formula}
    \hat{\bm{y}}(\bm{x}) = \hat{\bm{y}}(\bm{x}_0) + p\,(\bm{x} - \bm{x}_0) +\tfrac12 q(\bm{x}-\bm{x}_0,\bm{x}-\bm{x}_0) +  O(\bm{x} - \bm{x}_0)^3 \;,
\end{align}
where $p = D\hat{\bm{y}}(\bm{x}_0)$ denotes the $\ell \times n$ Jacobian, and the bilinear form $q = D^2\hat{\bm{y}}(\bm{x}_0)$ denotes the second derivative tensor of $\hat{\bm{y}}$ at $\bm{x}_0$. Taylor expanding the PDE system \Eq{general_pfaffian_pde} about $\bm{x}$ leads to a sequence of conditions on the coefficients in \Eq{Taylor_Formula}. At leading order we find an affine equation for $p$:
\begin{align} \label{eq:first-order-pfaff}
    (\alpha^a_{\bm{x}})^T + (\alpha^a_{\bm{y}})^Tp = 0,\quad a=1,\dots, r \;.
\end{align}
At first-order we find an equation that involves both $p$ and $q$. Equality of mixed partials $q(\delta\bm{x}_1 ,\delta\bm{x}_2) = q(\delta\bm{x}_2,\delta\bm{x}_1)$ allows us to eliminate $q$, and implies that 
\begin{equation}\label{eq:pfaff_compatable}
    \left[\dmat^a_{\bm{x}\bm{x}} + \dmat^a_{\bm{x}\bm{y}}p 
    +p^T\,\dmat^a_{\bm{y}\bm{x}} + p^T(\dmat^a_{\bm{y}\bm{y}})\,p \right] - 
    \left[ \ldots \right]^T = 0 \;, \quad a \in 1,\ldots r \; ;
\end{equation}
that is, the $n \times n$ matrix in $[\;]$ above is symmetric for each $a$. Here we have introduced the useful shorthand notation
\begin{equation}\label{eq:Dxx_Notation}
    \dmat^a_{\bm{x}\bm{x}} = D_{\bm{x}}\alpha^a_{\bm{x}}(\bm{x},\bm{y}),\quad \dmat^a_{\bm{x}\bm{y}} = D_{\bm{y}}\alpha^a_{\bm{x}}(\bm{x},\bm{y}),\quad \dmat^a_{\bm{y}\bm{x}} = D_{\bm{x}}\alpha^a_{\bm{y}}(\bm{x},\bm{y}),\quad \dmat^a_{\bm{y}\bm{y}} = D_{\bm{y}}\alpha^a_{\bm{y}}(\bm{x},\bm{y}).
\end{equation}

A solution $p$ of \Eq{first-order-pfaff} and \Eq{pfaff_compatable} at a point in $M$ is an \textbf{integral element}. We denote the collection of all integral elements at $(\bm{x},\bm{y})$ as $\cV_{(\bm{x},\bm{y})}\subset \bR^{\ell\times n}$. Integral elements may also be understood geometrically as $n$-dimensional subspaces of $\bR^n\times\bR^\ell$ that project onto $\bR^n$ and on which the $1$-forms $\alpha^a$ and their exterior derivatives $d\alpha^a$ all vanish. Vanishing of the $1$-forms leads to \Eq{first-order-pfaff}, while vanishing of the exterior derivatives leads to \Eq{pfaff_compatable}, see \App{CK-forms-ver}.

Generally $\cV_{(\bm{x},\bm{y})}$ is a nonlinear space due to the term in \Eq{pfaff_compatable} quadratic in $p$. However, it is common \Eq{first-order-pfaff} and \Eq{pfaff_compatable} degenerate to an affine system for $p$, in which case $\cV_{(\bm{x},\bm{y})}$ is an affine subspace of $\bR^{\ell\times n}$.\footnote
{Generally, when one sets up an exterior differential system for a PDE it is often done so using the language of jet spaces. All exterior differential systems expressed in this way turn out to be linear Pfaffian systems. See  \cite[Example 6.1.4]{Ivey16} for a proof. However, this is not always true for an arbitrary Pfaffian system with an independence condition.} 
When this happens the Pfaffian system is said to be \textbf{linear}. The aformentioned degeneracy can be detected by first expressing the general solution of \Eq{first-order-pfaff} as $p = p_{*} + p^\circ$, where $p_*$ denotes a particular solution and $p^\circ$ denotes the undetermined homogeneous solution. When this expression is back substituted into \Eq{pfaff_compatable} what remains is a quadratic equation for $p^\circ$, in which the quadratic terms vanish for a linear Pfaffian system. A sufficient condition for linearity is that the $\ell \times \ell$ matrix $\dmat^a_{\bm{y}\bm{y}}$
be symmetric for each $a=1,\dots, r$. However, this condition is not necessary because, for example, we may have $\dmat^a_{\bm{y}\bm{y}} = \alpha^a_{\bm{y}}(\bm{x},\bm{y})\bm{w}^T$ for some $\bm{w}\in \bR^\ell$. A necessary and sufficient condition for linearity can be expressed in terms of the set of \textbf{transverse vectors} at a point $(\bm{x},\bm{y}) \in \bR^n \times \bR^\ell$, defined by
\begin{equation}\label{eq:transverseVectors}
    J^\circ_{(\bm{x},\bm{y})} = \{ (\bm{0},\delta\bm{y}) : (\alpha^a_{\bm{y}})^T\delta\bm{y} = 0, \; a = 1,\ldots, r\} \;.
\end{equation}
When the $\alpha^a_{\bm{y}}$ are independent, $\text{dim}\,J^\circ_{(\bm{x},\bm{y})}=\ell-r$. A Pfaffian system with independence condition is linear if and only if for each $(\bm{x},\bm{y})$ and each pair of transverse vectors $(\bm{0},\delta\bm{y}_1^\circ),(\bm{0},\delta\bm{y}_2^\circ)\in J^\circ_{(\bm{x},\bm{y})}$ we have
\begin{align}\label{eq:linearity_condition}
    (\delta\bm{y}_2^\circ)^T(\dmat^a_{\bm{y}\bm{y}} - (\dmat^a_{\bm{y}\bm{y}})^T)(\delta\bm{y}_1^\circ) = 0,\quad a = 1,\dots, r.
\end{align}
This condition corresponds to vanishing of a certain $(\ell - r)\times (\ell - r)$ block of $\dmat^a_{\bm{y}\bm{y}} - (\dmat^a_{\bm{y}\bm{y}})^T$ for each $a$.

\subsection{Tableau of a linear Pfaffian system\label{sec:Pfaffian_Tableau}}

Every linear Pfaffian system with independence condition gives rise to a family of tableau as defined in \Sec{Linear_PDEs}. These tableau usefully encode many algebraic properties of the PDE system associated with the Pfaffian system. In this context it is conceptually helpful to think of tableau abstractly as subspaces of rectangular matrices; the interpretation as solutions of simple PDE systems can always be recovered when convenient for computations.

To do this, we first introduce the space of \textbf{horizontal vectors} at a point $(\bm{x},\bm{y}) \in \mathbb{R}^n\times \mathbb{R}^\ell$:
\[
    H_{(\bm{x},\bm{y})} = \big\{ (\delta\bm{x},\delta\bm{y}) \in \bR^n \times \bR^\ell:
        \delta\bm{y} = \sum_{a=1}^r c_a \alpha^a_{\bm{y}}, \text{ and }
        (\alpha^a_{\bm{x}})^T\delta\bm{x} + (\alpha^a_{\bm{y}})^T\delta\bm{y} = 0, \; a = 1,\ldots,r \big\} \;,
\]
By linear independence of the $\alpha^a_{\bm{y}}$ we have $\text{dim}\,H_{(\bm{x},\bm{y})} = n$. A basis for $H_{(\bm{x},\bm{y})}$ can be constructed using the Gram matrix $G(\bm{x},\bm{y})$, the $r\times r$  matrix with elements $G_{ab}(\bm{x},\bm{y}) = (\alpha^a_{\bm{y}})^T\alpha^b_{\bm{y}}$. Linear independence of the $\alpha_{\bm{y}}^a$ implies $G(\bm{x},\bm{y})$ is invertible. 
Using the inverse, define $\ell\times n$ matrix
\begin{align}\label{eq:gamma_matrix}
    \Gamma \equiv \sum_{a,b=1}^r\alpha_{\bm{y}}^a G^{-1}_{ab} (\alpha^b_{\bm{x}})^T.
\end{align}
It is not hard to see that a basis for $H_{(\bm{x},\bm{y})}$ is then given by
\begin{equation}\label{eq:H-basis}
    \varepsilon_i = (e_i,-\Gamma(\bm{x},\bm{y})e_i) \;, \quad i=1\dots, n \;.
\end{equation}
\
The tableau $\Pi \subset \bR^{r\times n}$ at each point $(\bm{x},\bm{y})$ is defined as the image of a certain linear map $\hat{\pi}:J^\circ_{(\bm{x},\bm{y})}\rightarrow \Pi$ that sends each transverse vector $\delta m^\circ = (0,\delta\bm{y}^\circ)$ to an $r\times n$ matrix $\hat{\pi}(\delta m^\circ)$ with entries $\pi^a_i(\delta m^\circ)$.
Note that each $\pi^a_i$ is a $1$-form on the vector space $J^\circ_{(\bm{x},\bm{y})}$, so that that $\hat{\pi}$ is an $r\times n$ matrix of $1$-forms. The entries $\pi^a_i(\delta m^\circ)$ may be defined simply in terms of the exterior derivatives of $\alpha^a$ as
\begin{align*}
    \pi^a_i(\delta m^\circ) = d\alpha^a(\delta m^\circ,\varepsilon_i),\quad a = 1,\dots, r,\quad i=1,\dots, n \;,
\end{align*}
with $\varepsilon_i$ in \Eq{H-basis}.
Differentiating the one-form \Eq{alphaxy} to obtain $d\alpha^a$, and using the notation \Eq{Dxx_Notation} then gives
\begin{equation}\label{eq:tableau_entries_general}
\begin{split}
    \pi^a_i(\delta m^\circ) &=(\delta \bm{y}^\circ)^T\Sigma^ae_i, \\
    \Sigma^a &\equiv (\dmat^a_{\bm{y}\bm{y}} - (\dmat^a_{\bm{y}\bm{y}})^T)\Gamma - (\dmat^a_{\bm{y}\bm{x}} - (\dmat^a_{\bm{x}\bm{y}})^T), \\
\end{split}
\end{equation}
Note that the image under $\hat{\pi}$ of any basis for the $\ell -r$ dimensional space $J^\circ$, spans $\Pi$, but may be ``over-complete" if $\hat{\pi}$ has a null space.

\subsection{Torsion of a linear Pfaffian system\label{sec:Pfaffian_Torsion}}

If the space of integral elements $\cV_{(\bm{x},\bm{y})}$ for \Eq{first-order-pfaff}--\Eq{pfaff_compatable} is non-empty for all $(\bm{x},\bm{y}) \in \mathbb{R}^n\times \mathbb{R}^\ell$ the \textbf{torsion} is said to vanish. If this is not the case,
then there is no Taylor series solution of \Eq{general_pfaffian_pde} at $(\bm{x},\bm{y})$ and we say the Pfaffian system \underline{has torsion} at $(\bm{x},\bm{y})$. In this case, the original PDE system must be amended to include the vanishing torsion condition. 

The vanishing of torsion can be formulated explicitly for a linear Pfaffian system as follows. By definition, linearity requires that \Eq{first-order-pfaff} and \Eq{pfaff_compatable} are affine. To clarify the affine structure, first notice that \Eq{first-order-pfaff} has the general solution \[p = -\Gamma + p^\circ,\] where $\Gamma$ is defined in \Eq{gamma_matrix} and $p^\circ$ is an $\ell\times n$ matrix whose columns $p_i$, $i=1,\dots, n$, are orthogonal to the $\alpha^a_{\bm{y}}$. Note that for each column of $p$ we have $(\bm{0},p_i)\in J^\circ_{(\bm{x},\bm{y})}$.
Substituting this expression for $p$ in \Eq{pfaff_compatable} and using \Eq{linearity_condition} then leads to
\begin{align}\label{eq:pfaffian_affine_eqn}
    (\Sigma^a)^Tp^\circ - (p^\circ)^T\Sigma^a= C^a,\quad a=1,\dots, r \;,
\end{align}
where $\Sigma^a$ is defined in \Eq{tableau_entries_general} and the skew-symmetric $n\times n$ matrix $C^a$ is given by
\begin{align}\label{eq:Ca_def}
    C^a =& -(\dmat^a_{\bm{x}\bm{x}} - (\dmat^a_{\bm{x}\bm{x}})^T) - \Gamma^T(\dmat^a_{\bm{y}\bm{y}} - (\dmat^a_{\bm{y}\bm{y}})^T)\Gamma\nonumber\\
    &+\Gamma^T(\dmat^a_{\bm{y}\bm{x}} - (\dmat^a_{\bm{x}\bm{y}})^T) - (\dmat^a_{\bm{y}\bm{x}} - (\dmat^a_{\bm{x}\bm{y}})^T)^T\Gamma \;.
\end{align}
It is straightforward to show that the entries of $C^a$ are related to the exterior derivative $d\alpha^a$ according to $C^a_{ij} = d\alpha^a(\varepsilon_i,\varepsilon_j)$, where $\varepsilon_i$ denotes the $i^{\text{th}}$ basis vector for the horizontal subspace $H_{(\bm{x},\bm{y})}$. Equation \Eq{pfaffian_affine_eqn}, which is equivalent to the combination of \Eq{first-order-pfaff} and \Eq{pfaff_compatable}, is now clearly affine. 

Let $(\bR^{\ell\times n})^\circ\subset \bR^{\ell\times n}$ denote the subspace of $\ell\times n$ matrices with columns orthogonal to each of the $\alpha^a_{\bm{y}}(\bm{x},\bm{y})$. Note that $p^\circ\in (\bR^{\ell\times n})^\circ$. Let $W$ denote the space of $r$-component column vectors whose components are skew-symmetric $n\times n$ matrices. Note that $C = (C^1,\dots, C^r)^T$ is an element of $W$. It is suggestive to write \Eq{pfaffian_affine_eqn} in terms of the linear map $\hat{\Sigma}:(\bR^{\ell\times n})^\circ\rightarrow W$ 
as
\begin{equation}\label{eq:hat_sigma_def}
    \hat{\Sigma}p^\circ = C \;. 
\end{equation} 
It is then clear that torsion vanishes at $(\bm{x},\bm{y})$ if and only if $C$ lies in the image of $\hat{\Sigma}$. 
Vanishing torsion is equivalent to existence of at least one integral element for each $(\bm{x},\bm{y})$, which is a necessary condition for existence of any solution of the PDE system \Eq{general_pfaffian_pde}.

\subsection{Brief recap}
To wrap-up this section, we end on a more formal statement of the Cartan-K\"ahler theorem/algorithm for linear Pfaffian systems. This is not the usual formulation of the theorem and we refer the interested reader to the standard literature, specifically Theorem 2.2 of  \cite{EDS-Book} and chapter sections 8.3 and 8.4 of  \cite{Ivey16}. 
\begin{theorem}\label{thm:CK}
Given an analytic PDE system of $\ell$ functions in $n$ independent variables written as a linear Pfaffian system on some open set $U\subset\mathbb{R}^n\times \mathbb{R}^\ell$ such that: 
\begin{enumerate}[label={(\Alph*)}]
    \item\label{ZeroTorsion} torsion as calculated in \Sec{Pfaffian_Torsion} vanishes on $U$ and 
    \item\label{CartanTest} Cartan's Test from \Sec{Linear_PDEs} holds for the associated tableau constructed in \Sec{Pfaffian_Tableau} on $U$,
\end{enumerate}
then there exists an analytic solution $(\bm{x},\bm{y}(\bm{x}))$ to the PDE on some open set $V\subset U$ determined by solutions to a sequence of Cauchy problems with analytic initial data on $U$. Moreover, the calculation for item \ref{CartanTest} describes the amount of freedom available in the choice of analytic Cauchy data. 
\end{theorem}

In particular, item \ref{ZeroTorsion} is an integrability condition and is therefore necessary for solutions to exist at all, whereas item \ref{CartanTest} is a sufficiency condition for the existence of analytic solutions provided \ref{ZeroTorsion} holds. 

We remark that non-analytic versions of the Cartan-K\"ahler theorem do exist provided one assumes additional hypotheses. The most famous of these are those cases of involutive hyperbolic systems, first explored by  \cite{YangHyperbolicCK}. 
\section{Examples}\label{sec:Examples}

\subsection{First example: Simple PDEs as Pfaffian systems}\label{sec:SimplePfaffian}
As a first example, we show that the tableau associated in
\Sec{Linear_PDEs} with the simple PDE system in \Eq{SimplePDE} agrees with the tableau obtained by expressing the same system as a linear Pfaffian system and applying the construction of \Sec{Pfaffian_Tableau}. We note that Chapter 5 of  \cite{Ivey16} explores the general theory of equations of the type \Eq{SimplePDE} in both greater detail and additional generality of the underlying theory. 

If $e_i, i=1,\dots, n$ is the standard basis of $\bR^n$, then the tableau of \Eq{SimplePDE} is effectively the subspace of $d\times n$ matrices
\[ \Pi = \left\{A\in \bR^{d\times n} \,\left|\, \sum_{i=1}^n L^i A e_i = 0 \right.\right\}. \]
Set $m = \dim \Pi$ and let $A_1,\dots, A_m \in \Pi$ be a basis of the tableau.

Introduce coordinates $(\bm{x},\bm{u},\bm{w})$ on $\bR^n\times\bR^{d} \times \bR^m $ and set $\bm{y} = (\bm{u},\bm{w}) \in \bR^{\ell}$ with $\ell = d + m$. The simple PDE system is transformed to the Pfaffian system given by the $r = d$ set of 1-forms 
\begin{equation}\label{eq:simple_system_pfaffian_forms}
    \alpha^a = du^a - \sum_{i=1}^n \sum_{\epsilon=1}^m  A_{\epsilon i}^a w^\epsilon \,dx^i, 
      \qquad a=1,\dots,d.
\end{equation}

To see that this system does indeed cast the simple PDE \Eq{SimplePDE} as a Pfaffian system, observe that, for each $a = 1,\dots, d$, we can identify the components in \Eq{alphaxy} as
\[ 
    \left(\alpha_{\bm{x}}^a\right)_i = - \sum_{\epsilon = 1}^m  A_{\epsilon i}^a w^\epsilon, \quad i=1,\dots,n,\qquad 
    \left(\alpha^a_{\bm{y}}\right)_k = \delta_k^a, \quad k=1,\dots, \ell \;,
\]
for the Kronecker delta, $\delta_k^a$. 
Then, the PDE that corresponds to the Pfaffian system is, per \Eq{general_pfaffian_pde},
\[ 
    D \bm{u}(\bm{x}) = \sum_{\epsilon = 1}^m  A_{\epsilon} w^\epsilon(\bm{x}) \;.
\]
This shows that $D\bm{u}(\bm{x}) \in \Pi$, thus $\bm{u}(\bm{x})$ is a solution to the simple PDE system.

Let the $\ell\times n$ matrix $p$ given in \Eq{first-order-pfaff} be $ p = \begin{pmatrix} p_{\bm{u}} \\ p_{\bm{w}}\end{pmatrix}$ with $p_{\bm{u}}\in \bR^{d\times n}$ and $p_{\bm{w}} \in \bR^{m \times n}$. The leading order equation, \Eq{first-order-pfaff}, is simply
\[
   p_{\bm{u}} = \sum_{\epsilon = 1}^m w^\epsilon A_{\epsilon} \;,
\] 
while $p_{\bm{w}}$ is initially undetermined. The compatibility equation, \Eq{pfaff_compatable}, becomes 
\[ 
    \sum_{\epsilon=1}^m \left( A_{\epsilon i}^a (p_{\bm{w}})^\epsilon_{k} - A^a_{\epsilon k} (p_{\bm{w}})^\epsilon_{i} \right) = 0,\qquad k\neq i \;.
\]
Both the leading-order equations and the compatibility equations are affine linear in $p$. Hence \Eq{simple_system_pfaffian_forms} is a linear Pfaffian
system.
 
The transverse space is 
\[
    J^\circ_{(\bm{x},\bm{y})} = \left\{(0,0,\delta\bm{w}) \,|\, \delta\bm{w}\in \bR^m \right\} \;.
\]
The vectors $\alpha^1_{\bm{y}},\dots, \alpha^d_{\bm{y}}$ are linearly independent, thus, their Gram matrix is $G = I_d$. Then, the matrix $\Gamma$ defined in \Eq{gamma_matrix} is 
\[ 
    \Gamma = -\begin{pmatrix} \sum_{i=1}^\epsilon A_\epsilon w^\epsilon & 0_{m\times n} \end{pmatrix} \;.
\]
If $e_i$ is the standard basis of $\bR^n$, then the horizontal space $H_{(\bm{x},\bm{y})}$ is spanned by the vectors
\[
   \varepsilon_i= (e_i, \delta \bm{u}_i, 0),\qquad (\delta\bm{u}_i)_a = \sum_{\epsilon=1}^m A_{\epsilon i}^a w^\epsilon \;.
\]
Then the tableau (in the Pfaffian-systems-theory sense) is the image of the linear map $\hat{\pi}: J^\circ_{(\bm{x},\bm{y})} \to \bR^{d\times n}$ that sends each transverse vector $\delta m^\circ$ to a $d\times n$ matrix as defined in \Eq{tableau_entries_general}. 
If $\delta m^\circ = (0,0,\delta\bm{w}) $, we hence obtain the image of $\hat{\pi}$ as 
\[
    \pi_i^a(\delta m^\circ) = -A^a_{\epsilon i} \delta w^\epsilon \;.
\]
it follows that the image of $\hat\pi$ is precisely $\Pi$. This shows that the tableau associated with a simple PDE system according to \Sec{Linear_PDEs} agrees with the Pfaffian-systems-theory tableau described in \Sec{Pfaffian_Tableau}.


\subsection{Second example: A seemingly overdetermined system}\label{sec:CK-Example}

We now represent the simple example in \Sec{CK-CauchySeq} as a sequence of Cauchy problems. This example will also be used to demonstrate all the steps of the Cartan-K\"ahler algorithm with minimal use of differential forms. Note that the order of steps in \Sec{Cartan-Kahler} is not exactly as below. Step 2 of \Sec{Cartan-Kahler} has been combined with step 4, so that the ordering is now: 1) Encode PDE as Pfaffian System, 2) Check for vanishing torsion, and 3) find the tableau and check involutivity. Generally, when using differential forms, the steps are as outlined in \Sec{Cartan-Kahler}. See \App{CK-forms-ver} for this same example presented using the more compact approach of differential forms.

\subsubsection{Encode PDE as Pfaffian system}\label{sec:CK-Example-Pfaff}
We can rewrite the example PDE \Eq{CK-Example} as a Pfaffian system with $n=3$, $\ell = 10$, and $r=4$ of the form
\begin{equation}\label{eq:EDS-CK-Example}
\begin{aligned}
    \alpha^1&=du^1-p^1_1\,dx^1-(u^2-u^1+u^3u^4)\,dx^2-dx^3,\\
    \alpha^2&=du^2-p^2_1\,dx^1-(u^1-u^2)\,dx^2-dx^3,\\
    \alpha^3&=du^3-p^3_1\,dx^1-p^3_2\,dx^2-u^3\,dx^3,\\
    \alpha^4&=du^4-p^4_1\,dx^1-p^4_2\,dx^2+u^4\,dx^3. 
\end{aligned}
\end{equation}
The $\alpha^a$ should be interpreted as $1$-forms for  $(\bm{x},\bm{y}) \in \bR^3\times\bR^{10}$, where
\[
    \bm{x} = (x^1,x^2,x^3)^T,\quad 
    \bm{y} = (u^1,u^2,u^3,u^4,p^1_1,p^2_1,p^3_1,p^3_2,p^4_1,p^4_2)^T.
\]
Note that $p^1_{2}$, $p^1_{3}$, $p^2_{2}$, $p^2_{3}$, $p^3_3$ and $p^4_3$ have been replaced by the six expressions for $\partial_{x^i}u^a$ from the PDE system \Eq{CK-Example}. The coefficient vectors $\alpha^a_{\bm{x}}$ and $\alpha^a_{\bm{y}}$ are given by
\begin{align*}
    \alpha^1_{\bm{x}} &= -(p^1_1,u^2-u^1 + u^3\,u^4,1)^T\\
    \alpha^2_{\bm{x}} & = -(p^2_1,u^1-u^2,1)^T\\
    \alpha^3_{\bm{x}} & = -(p^3_1,p^3_2,u^3)^T\\
    \alpha^4_{\bm{x}} & = -(p^4_1,p^4_2,-u^4)^T,
\end{align*}
and
\[
    \alpha^1_{\bm{y}}  = e_1,\quad \alpha^2_{\bm{y}} =e_2,\quad \alpha^3_{\bm{y}} = e_3,\quad \alpha^4_{\bm{y}} = e_4,
\]
where $e_\mu$, $\mu=1,\dots, 10$, denotes the standard basis for $\bR^{10}$. 
Note that this Pfaffian system satisfies an independence condition because the $\alpha^a_{\bm{y}}$ are linearly independent. The matrices defined by \Eq{Dxx_Notation} are given by
\begin{gather*}
    \dmat^a_{\bm{x}\bm{x}}  = 0,\quad \dmat^a_{\bm{y}\bm{x}} = 0,\quad \dmat^a_{\bm{y}\bm{y}} =0,\quad   a = 1,\dots, 4\\
    \dmat^1_{\bm{x}\bm{y}} = -e_1\,e_5^T - e_2\,(e_2-e_1 + u^3\,e_4 + u^4\,e_3)^T,\quad \dmat^2_{\bm{x}\bm{y}} = -e_1\,e_6^T - e_2\,(e_1-e_2)^T\\
    \dmat^3_{\bm{x}\bm{y}} = -e_1\,e_7^T -e_2\,e_8^T - e_3\,e_3^T,\quad \dmat^4_{\bm{x}\bm{y}} = -e_1\,e_9^T - e_2\,e_{10}^T + e_3\,e_4^T.
\end{gather*}

Thus this Pfaffian system is linear by \Eq{linearity_condition} because $\dmat^a_{\bm{y}\bm{y}} = 0$.

\subsubsection{Check if torsion vanishes}\label{sec:CK-Example-Torsion}
To check if torsion vanishes we must check if the vector of skew-symmetric matrices $C^a$ defined in \Eq{Ca_def} lie in the image of $\hat{\Sigma}$ defined in \Eq{hat_sigma_def}.

We first note that $G = I_{4 \times 4}$  so that from \Eq{gamma_matrix}, the $10 \times 3$ matrix $\Gamma$ is explicitly 
\begin{align*}
    \Gamma =& -\begin{pmatrix}
         p^1_1 & u^2-u^1 + u^3 u^4 &1 \\
         p^2_1 &  u^1-u^2 &  1\\
         p^3_1 &  p^3_2 &  u^3\\
         p^4_1 &  p^4_2 &- u^4 \\
         0 & 0 & 0 \\
         & \ldots  & 
        \end{pmatrix} \;,
\end{align*}
omitting the final five rows of zeros.
By \Eq{Ca_def} and the above formulas for the $\dmat$-matrices, we have $C^a = -\Gamma^T(\dmat^a_{\bm{x}\bm{y}})^T + \dmat^a_{\bm{x}\bm{y}}\Gamma$, or explicitly
\begin{align*}
    C^1 &= \begin{pmatrix} 0 & p^1_1 - p^2_1 - p^4_1 u^3 - p^3_1 u^4 & 0\\  -p^1_1 + p^2_1 + p^4_1 u^3 + p^3_1 u^4 &  0 & 0\\ 0 & 0 & 0\end{pmatrix}\\
    C^2 & = \begin{pmatrix} 0 & p^2_1 - p^1_1 & 0\\ -(p^2_1 - p^1_1) & 0\\ 0 & 0 & 0\end{pmatrix},\quad C^3 = \begin{pmatrix} 0 & 0 & -p^{3}_1\\ 0 & 0 & -p^3_2\\ p^3_1 & p^3_2 & 0\end{pmatrix},\quad C^4 = \begin{pmatrix} 0 & 0 & p^4_1\\ 0 & 0 & p^4_2 \\ -p^4_1 & -p^4_2 & 0\end{pmatrix}.
\end{align*}
The matrices $\Sigma^a$ defined in \Eq{tableau_entries_general} become simply $\Sigma^a = (\dmat^a_{\bm{x}\bm{y}})^T$, giving
\begin{align*}
    \Sigma^1 = \begin{pmatrix} 
        0 &1 &0\\0 &-1 &0\\0 &-u^4 &0\\
        0 &-u^3 &0\\-1 &0 &0\\0 &0 &0\\
        0 &0 &0\\0 &0 &0\\0 &0 &0\\0 &0 &0\\ 
    \end{pmatrix},\,
    \Sigma^2 = \begin{pmatrix} 
        0 &-1 &0\\0 &1 &0\\0 &0 &0\\
        0 & 0 &0\\0 &0 &0\\-1 &0 &0\\
        0 &0 &0\\0 &0 &0\\0 &0 &0\\0 &0 &0\\ 
    \end{pmatrix},\,
    \Sigma^3 = \begin{pmatrix}
        0 &0 &0\\0 &0 &0\\0 &0 &-1\\
        0 &0 &0\\0 &0 &0\\0 &0 &0\\
        -1 &0 &0\\0 &-1 &0\\0 &0 &0\\0 &0 &0\\ 
    \end{pmatrix},\,
    \Sigma^4 = \begin{pmatrix}
        0 &0 &0\\0 &0 &0\\0 &0 &0\\
        0 &0 &1\\0 &0 &0\\0 &0 &0\\
        0 &0 &0\\0 &0 &0\\-1 &0 &0\\0 &-1 &0\\ 
    \end{pmatrix} .
\end{align*}
To find the image of $\hat{\Sigma}$ first notice that its domain space is $(\bR^{\ell\times n})^\circ$, which comprises  $\ell\times n$ matrices $p^\circ$ with columns that are each orthogonal to the $\alpha^a_{\bm{y}}$. By the above formulas for $\alpha^a_{\bm{y}}$, the most general $p^\circ\in (\bR^{\ell\times n})^\circ$ has the form
\begin{align*}
    p^\circ = \begin{pmatrix} 0 & 0 &0\\0 & 0 &0\\0 & 0 &0\\0 & 0 &0\\ v^5_1 & v^5_2 & v^5_3\\v^6_1 & v^6_2 & v^6_3\\v^7_1 & v^7_2 & v^7_3\\v^8_1 & v^8_2 & v^8_3\\v^9_1 & v^9_2 & v^9_3\\v^{10}_1 & v^{10}_2 & v^{10}_3\\\end{pmatrix}
\end{align*}
By \Eq{hat_sigma_def} the components of $\hat{\Sigma}p^\circ$ are therefore
\begin{equation}\label{eq:SigmaOperator}
\begin{aligned}
    (\hat{\Sigma}p^\circ)^1  &= \begin{pmatrix}0 & -v^5_2 & -v^5_3\\ v^5_2 & 0 & 0\\ v^5_3 & 0 & 0 \end{pmatrix},\quad\quad\quad 
   (\hat{\Sigma}p^\circ)^2  = \begin{pmatrix}0 & -v^6_2 & -v^6_3\\ v^6_2 & 0 & 0\\ v^6_3 & 0 & 0 \end{pmatrix} ,\\
    (\hat{\Sigma}p^\circ)^3 &= \begin{pmatrix}0 & v^8_1 - v^7_2 & -v^7_3\\ v^7_2 - v^8_1 & 0 & -v^8_3\\ v^7_3 & v^8_3 & 0\end{pmatrix}, 
    (\hat{\Sigma}p^\circ)^4 = \begin{pmatrix}0 & v^{10}_1 - v^9_2 & -v^{9}_3\\ v^9_2 - v^{10}_1 & 0 & -v^{10}_3\\ v^9_3 & v^{10}_3 & 0 \end{pmatrix} \;.
\end{aligned}
\end{equation}
These formulas reveal that $C$ does indeed lie in the image of $\hat{\Sigma}$. We find that a particular solution of $\hat{\Sigma}p^\circ = C$ is
\begin{align}\label{eq:Sigma-hat_image_ex}
    v^5_2 & = p^2_1 - p^1_1 + p^4_1\,u^3 + p^3_1\,u^4,\quad v^5_3 = 0 \;,\\
    v^6_2 & = p^1_1 - p^2_1,\quad v^6_3 = 0 \;,\\
    v^7_2 &= 0,\quad v^7_3 = p^3_1,\quad  v^8_1 = 0,\quad v^8_3 = p^3_2 \;,\\
    \quad v^9_2 &= 0,\quad v^9_3 =  - p^4_1,\quad v^{10}_1  = 0,\quad v^{10}_3 = -p^4_2 \;.
\end{align}
It follows that the torsion vanishes.

\subsubsection{Check if tableau is involutive\label{sec:check_tableau_example}}
First we compute the tableau of this linear Pfaffian system with independence condition using \Eq{tableau_entries_general}. For each $(\bm{x},\bm{y})\in \mathbb{R}^3\times\mathbb{R}^{10}$ the tableau $\Pi$ will be a subspace of $r\times n = 4\times 3$ matrices given by the image of the mapping $\hat{\pi}$ that sends transverse vectors $\delta m^\circ = (0,\delta\bm{y}^\circ)$ to $4\times 3$ matrices. The most general transverse vector has
\begin{align*}
    \delta\bm{y}^0 = \begin{pmatrix}0 & 0 & 0& 0& v^5& v^6& v^7& v^8 & v^9& v^{10} \end{pmatrix}^T \;.
\end{align*}
The $4\times 3$ matrix that results from applying $\hat{\pi}$ to $\delta m^\circ$ is therefore
\begin{align*}
    \hat{\pi}(\delta m^\circ)  = 
        \begin{pmatrix} -v^5 & 0& 0\\-v^6 & 0& 0\\
                        -v^7 & -v^8& 0\\-v^9 & -v^{10}& 0
        \end{pmatrix} \;.
\end{align*}
The most tedious part of computing $\hat{\pi}$ is finding the $\Sigma^a$ matrices and the general form of the transverse vectors, each of which was already accomplished when computing the torsion. It follows that the tableau is given by
\begin{align*}
    \Pi = \left\{\left. \begin{pmatrix} -v^5 & 0& 0\\-v^6 & 0& 0\\-v^7 & -v^8& 0\\-v^9 & -v^{10}& 0\end{pmatrix}\ \right|\ v^i\in \bR,\quad i=5,\dots, 10 \right\}.
\end{align*}
This tableau can be viewed equivalently as the space of homogeneous linear solutions of the simple PDE system for $\hat{\bm{w}}:\bR^3\rightarrow \bR^4$ given by
\begin{equation*}
\begin{split}
    \partial_{x^2} \hat{w}_1 &= 0 \;,\\
    \partial_{x^2} \hat{w}_2 &= 0 \;,\\  
\end{split}
\quad
\begin{split}
    \partial_{x^3}\hat{w}_1 &= 0 \;,\\
    \partial_{x^3}\hat{w}_2 &= 0 \;,\\
    \partial_{x^3}\hat{w}_3 &= 0 \;.
\end{split}
\end{equation*}

Notice that this simple system is very different from the original PDE system \Eq{CK-Example}; this is to be expected more generally.

Next we apply Cartan's test to assess involutivity of the tableau, as described in \Sec{Linear_PDEs}. 

The subspaces $\Pi_i$ are given by
\begin{align*}
    \Pi_0 = \Pi,\quad \Pi_1 = \left\{\left. \begin{pmatrix} 0 & 0& 0\\0 & 0& 0\\0 & -v^8& 0\\0 & -v^{10}& 0\end{pmatrix}\ \right|\ v^8,v^{10}\in\bR \right\},\quad \Pi_2 = \left\{ \begin{pmatrix} 0 & 0& 0\\0 & 0& 0\\0 & 0& 0\\0 & 0& 0\end{pmatrix} \right\},\quad \Pi_3 = \Pi_2 \;.
\end{align*}
The Cartan characters are therefore
\begin{align*}
    s_1 = 4,\quad s_2 = 2,\quad s_3 = 0 \;.
\end{align*}
In order to now check the inequality \Eq{cartan_test_inequality} we need the dimension of the first prolongation of the tableau $\Pi^{(1)}$. While it is straightforward to compute the prolongation directly, there is a convenient shortcut that reuses some of the work from computing the torsion. The dimension of the first prolongation of the tableau for a linear Pfaffian system with independence condition is always given by the dimension of the null space for $\hat{\Sigma}$. Straightforward inspection of \Eq{SigmaOperator} shows $\text{dim}\,\text{ker}\,\hat{\Sigma} = 8$. It follows that $\text{dim}\,\Pi^{(1)} = 8$. Finally we conclude that Cartan's test passes because $s_1 + 2s_2 + 3 s_3 = 8 = \text{dim}\,\Pi^{(1)}$. The tableau is therefore involutive for all $(\bm{x},\bm{y})\in \mathbb{R}^3\times\mathbb{R}^{10}$.

In summary, since torsion vanishes and the tableau is involutive the Cartan-K\"ahler algorithm now terminates. Since the original PDE system was real analytic, we infer the existence of local analytic solutions near any point in $\bm{x}$-space. Of course, obtaining the solutions to this system is straightforward, as seen in \App{Example}.

\section{Realizability as an exterior differential system}\label{sec:BoozerEDS}

In this section and those that follow, we prove \Th{FTSD}: our goal is to show that, given a magnitude $\cB$ on an open set $Q \subset \bR^3$ and flux functions $F,G,K,L$ with $\Delta = FL-GK > 0$, there is a (local) solution to the PDE system \Eq{BoozerPDE}. In order to apply the Cartan-K\"ahler algorithm described in \Sec{Cartan-Kahler} to this problem, we will reformulate the PDE system \Eq{BoozerPDE} as a linear Pfaffian system with independence condition (cf \Sec{Pfaffian_Defined}). The process of reformulating the PDE as a linear Pfaffian system is done in two parts. First, in this section, we establish the realizability problem as an exterior differential system (EDS) (see \App{CK-forms-ver} for definitions of EDS). Then, we  `linearize' the EDS using prolongation to produce a linear Pfaffian system in \Sec{ProlongEDS}. Finally, we carry out the Cartan-K\"ahler algorithm in \Sec{Proof}.

More concretely, in this section, we establish \Lem{EDSReal} which recasts the original PDE system \Eq{BoozerPDE} as an EDS by using the straightforward differential forms version of realizability given in \Eq{BoozerDiffForm}. However, this formulation of the EDS is not easily linearized into a Pfaffian system (as we will want to do in \Sec{ProlongEDS}). Therefore, in \Lem{SphereEDS} we reformulate the realizability problem as a second, equivalent EDS that can be efficiently linearized. This efficient EDS is a result of a geometrically-inspired choice of `coordinates' (more technically, a choice of \textit{coframe}). The details of this coframe are given primarily in \App{Convenient_Frame} and recorded in \Lem{SphereBundle}.

We begin by representing the realizability problem as an exterior differential system $\cI^0$ on the 8D manifold $M =SQ\times A$ where
\begin{equation}\label{eq:SphereBundle}
    SQ = \{(\bm{x},\bm{b}) \;|\; \bm{x} \in Q ,\; \bm{b} \in \bR^3 ,\; |\bm{b}| = 1\}
\end{equation}
is the unit sphere bundle over $Q$. In particular, for the coordinates $(\bm{x},\bm{b},\psi,\theta,\zeta) \in SQ\times A$, the integral $3$-manifolds of this EDS will be of the form $(\bm{x}, \bm{b}(\bm{x}), \psi(\bm{x}), \theta(\bm{x}), \zeta(\bm{x}))$.

\begin{lemma}[Straightforward EDS Form of Realizability]\label{lem:EDSReal}
    For a given $\cB$ and flux functions $F,G,K,L$, a function $\bm{x} \mapsto (\bm{B}(\bm{x}),\psi(\bm{x}), \theta(\bm{x}), \zeta(\bm{x}))$ is a solution to \Eq{BoozerPDE} if and only if $(\bm{x}, \bm{b}(\bm{x}),\psi(\bm{x}), \theta(\bm{x}), \zeta(\bm{x}))$ is an integral 3-manifold for the EDS $\cI^0$ on $M = SQ\times A$ generated by the 2-forms
    \begin{equation}\label{eq:EDS_I}
    \begin{aligned}
        \beta^1 &= \cB(\psi,\theta,\zeta)\,\iota_{\bm{b}}\Omega - 
        \left(F(\psi)\,d\psi\w  d\theta - G(\psi)\,d\psi\w  d\zeta \right) \;,\\
        \beta^2 & = \cB(\psi,\theta,\zeta)\,\bm{b}^\flat\w  d\psi - 
        \left(K(\psi)\,d\psi\w  d\theta - L(\psi)\,d\psi\w  d\zeta \right) \;,
    \end{aligned}
    \end{equation}
    where $\Omega$ is the Euclidean volume form on $Q$. 
\end{lemma}
\begin{proof}
   Recall the well-known relationship between differential forms and vector calculus established by a metric,

   namely, for a function $h$ and vectors $X,Y$ we have $dh^\sharp = \nabla h$, $X^\flat = \star\iota_X \Omega$, and $(X\times Y)^\flat = \star(X^\flat\wedge Y^\flat)$ \cite{MacKay20a}.
   Using these relations, a calculation shows that a function $\Phi:\bm{x} \mapsto (\psi(\bm{x}), \theta(\bm{x}), \zeta(\bm{x}))$ is a solution to \Eq{BoozerPDE} if and only if 
\begin{align*}
    f &= |\bm{B}|^2 - \cB^2(\psi,\theta,\zeta) \;,\\
    \beta^1 & = \iota_{\bm{B}}\Omega - \left(F(\psi)\,d\psi\w  d\theta -
         G(\psi)\,d\psi\w  d\zeta \right) \;,\\
    \beta^2 & = \bm{B}^\flat \w  d\psi - \left(K(\psi)\,d\psi\w  d\theta - 
          L(\psi)\,d\psi\w  d\zeta \right) \;,
\end{align*}
vanish when pulled back to the 3D submanifold given by the graph $(\bm{x}, \bm{B}(\bm{x}),\psi(\bm{x}), \theta(\bm{x}), \zeta(\bm{x}))$. That is, the projectable, integral 3-manifolds of the EDS $\langle f, \beta^1,\beta^2\rangle$ are in correspondence with solutions to \Eq{BoozerPDE}. 

Since the function $f$ in the above EDS must vanish, every integral manifold must lie within the submanifold $f^{-1}(\{0\})$ where $\bm{B} = \cB \bm{b}$. This is diffeomorphic to the 8D manifold $M =SQ\times A$ for the sphere bundle \Eq{SphereBundle}. We may therefore replace the EDS above with the slightly smaller EDS, $(M,\cI^0)$ given in \Eq{EDS_I}.
\end{proof}

Given the EDS \Eq{EDS_I}, we now construct a coframe on $M$ that greatly simplifies the calculations involved in the Cartan algorithm. Generally, a \textbf{coframe} on a manifold $M$ is a basis set of $1$-forms on its cotangent bundle $T^*M$. The current basis is effectively  $\{d\bm{x},(\bm{b}^\perp)^\flat,d\psi,d\theta,d\zeta\}$, where $(\bm{b}^\perp)^\flat$ represents a pair of linearly independent differential $1$-forms annihilating the space normal to $S^2$ in the sphere bundle. 

To formulate the new coframe, we first choose new
$1$-forms $\mu^1,\mu^2$ to represent the angles $(\theta, \zeta)$ by defining
\begin{equation} \label{eq:muDef}
\begin{aligned}
    \mu^1& = F(\psi)\,d\theta-G(\psi)\,d\zeta \;,\\
    \mu^2& = K(\psi)\,d\theta-L(\psi)\,d\zeta \;.
\end{aligned}
\end{equation}

The inverse relation is 
\begin{align*}
    d\theta &= \tfrac{1}{\Delta} (L\mu^1-G\mu^2) \;, \\
    d\zeta  &= \tfrac{1}{\Delta} (K\mu^1-F\mu^2) \;.
\end{align*}
Note the inverse is well-defined since $\Delta=FL-GK$ is assumed to be positive in \Th{FTSD}, recall \Eq{positivity}. Hence, $\mu^1,\mu^2$ can replace $d\theta,d\zeta$ in the coframe. 

Finally, we further simplify calculations by adapting the coframe to the direction of $\bm{b}$. Doing so requires a bit of theoretical framework; see \App{Convenient_Frame} for the rigorous derivation. The key idea is to appropriately establish $SQ$ as a subbundle of the full, (oriented) 6-dimensional orthonormal frame bundle $FQ$ over $Q$. This can be done for any set $U\subset SQ$ such that there is an analytic map---a \textit{choice of gauge}---$\cG:U \to FQ$ of the form $\cG(\bm{x}, \bm{b}) = (\cG_1(\bm{x}, \bm{b}), \cG_2(\bm{x},\bm{b}), \bm{b})$ where the output is an orthonormal basis of $T_{\bm{x}}Q$. Then, the frame bundle $FQ$ has a natural coframe given by the 3-components $\omega^1,\omega^2,\omega^3$ of the \textit{tautological 1-form $\omega$} (also known as the soldering or fundamental 1-form) together with the three connection 1-forms $\omega^3_2, \omega^3_1, \omega^1_2$. Finally, the map $\cG$ can be used to pull-back the coframing of $FQ$ to a coframing of $U$ which inherits the nice properties from $FQ$ (see \Lem{fExists}). In particular, the pullbacks of the $\omega^1, \omega^2,\omega^3,\omega^3_1,\omega^3_2$ form a coframe of $U\subset SQ$. For a general choice of gauge, the pullback $\cG^*\omega^1_2$ is dependent on all five coframe 1-forms. However, for the Euclidean metric $g$, there is a choice of gauge so that 
\begin{equation}\label{eq:omega12_Dependence}
    \omega^1_2 = w_1\, \omega^3_1 + w_2\, \omega^3_2
\end{equation}
for functions $w_1,w_2 :U \to \bR$, and where we have suppressed the pullback by $\cG$ notation for readability. 

For the Euclidean metric, the tautological and connection 1-forms have a straightforward description in terms of the 1-forms $d\bm{x}$ and $d\bm{b}$. Indeed, consider a choice of $\operatorname{SO}(3)$ matrices analytic in $(\bm{x},\bm{b})$ 
\[ \cG(\bm{x},\bm{b}) = \begin{pmatrix}
    \cG_1(\bm{x},\bm{b}) & \cG_2(\bm{x},\bm{b}) & \bm{b}
\end{pmatrix}. \]
Then the tautological 1-form is given explicitly as
\[ \omega = \cG^{T} d\bm{x},\quad\text{ that is,}\quad \begin{pmatrix}
    \omega^1 \\ \omega^2 \\ \omega^3
\end{pmatrix} = \cG^T \begin{pmatrix}
    dx^1 \\ dx^2 \\ dx^3
\end{pmatrix}.\]
In particular, $\omega^3 = b_1 dx^1 + b_2 dx^2 + b_3 dx^3 = \bm{b}^\flat$. The connection 1-forms $\omega^i_j$ are given in matrix form as 
\[ \begin{pmatrix}
    0 & \omega^1_2 & - \omega^3_1 \\
    -\omega^1_2 & 0 & -\omega^3_2 \\
    \omega^3_1 & \omega^3_2 & 0
\end{pmatrix} = \cG^T d\cG. \]
An explicit example is illuminating. Take as coordinates $(x^1,x^2,x^3,\theta,\phi) \in SQ$ with 
\[ \bm{b} = ( \sin\theta\cos\phi, \sin\theta\sin\phi,\cos\theta ).\]
Then a particular choice of gauge is the $\operatorname{SO}(3)$ matrix
\[ \cG = \begin{pmatrix}
    \cos\theta \cos\phi & -\sin\phi &  \sin\theta \cos\phi \\
    \cos\theta \sin\phi & \cos\phi & \sin\theta \sin\phi \\
    -\sin\theta & 0 & \cos\theta
\end{pmatrix}. \]
With this choice of gauge, the tautological 1-form is
\[ \begin{pmatrix}
    \omega^1 \\ \omega^2 \\ \omega^3
\end{pmatrix} = \cG^T d\bm{x} = \begin{pmatrix}
    \cos\theta \cos\phi\, dx^1 + \cos\theta \sin\phi\, dx^2 - \sin\theta\, dx^3 \\
    -\sin\phi\, dx^1 + \cos\phi\, dx^2 \\
    \sin\theta \cos\phi\, dx^1 + \sin\theta \sin\phi\, dx^2 + \cos\theta\, dx^3
\end{pmatrix} .\]
In turn, the connection 1-forms are given as
\[ \omega^3_1 = -d\theta,\qquad \omega^3_2 = -\sin\theta\, d\theta,\qquad \omega^1_2 = -\cos\theta\,d\phi. \]

In particular, $\omega^1_2 = \cot\theta\, \omega^3_2$, so that $w_1 = 0$ and $w_2 = \cot\theta$.

Henceforth, we let $U$ be an open set of $SQ$ with a gauge $\cG:U\to FQ$ that ensures \Eq{omega12_Dependence} holds. Let $M_U = U \times A$. The following lemma captures the desired coframing of $M_U$ and associated structure equations. 

\begin{lemma}[Sphere Bundle Coframe]\label{lem:SphereBundle}
    There is an open neighborhood $U$ of any point $(\bm{x},\bm{b})$ and choice of gauge $\cG:U \to FQ$  so that, on $M_U = U\times A$ the 1-forms $\{\omega^1,\omega^2,\omega^3, \omega^3_1,\omega^3_2,d\psi,\mu^1,\mu^2 \}$ form a coframe of $M$. Moreover, with $\eta^i$ defined as
    \begin{equation}\label{eq:eta-define}
        \eta \equiv (\omega^3_1,\omega^3_2, d\psi, \mu^1,\mu^2) \;,
    \end{equation}
    the coframe has structure equations
    
    \begin{equation}\label{eq:omegaStructure}
    \begin{aligned}
      d\omega^1 &= -(w_1\, \eta^1 + w_2\, \eta^2)\w \omega^2 + \eta^1 \w \omega^3 \;,\\
      d\omega^2 &= (w_1\, \eta^1 + w_2\, \eta^2) \w \omega^1 + \eta^2 \w \omega^3\;, \\
      d\omega^3 &= -\eta^1 \w \omega^1 - \eta^2 \w \omega^2 \;,
    \end{aligned}
    \end{equation}

    \begin{equation}\label{eq:fullStructure}
    \begin{aligned}
        d\eta^1&= - w_1\, \eta^1 \w \eta^2 \;,\\
        d\eta^2&= - w_2\, \eta^1 \w \eta^2 \;,\\
        d\eta^3 &= 0 \;, \\
        d\eta^4& = \tfrac{1}{\Delta}\left[ (F'L-G'K)\, \eta^3\w \eta^4 + (G'F-F'G)\, \eta^3 \w \eta^5 \right]\;,\\
        d\eta^5& = \tfrac{1}{\Delta}\left[{(K'L-L'K})\, \eta^3\w \eta^4 + (L'F-K'G)\, \eta^3 \w \eta^5 \right]\;,
    \end{aligned}
    \end{equation}
\end{lemma}
\begin{proof}
    The existence of $U$, gauge $\cG$, and linear independence of the coframe is established in \App{Convenient_Frame}. All that remains is to show is that the structure equations hold. The structure equations for $\omega^i$ and $\eta^1,\eta^2$ hold from \Eq{CartanCurvature} and the relation \Eq{omega12_Dependence}. The exterior derivatives of \Eq{muDef} give the structure equations
    \[
    \begin{aligned}
       d\eta^4 = d\mu^1& = \tfrac{1}{\Delta}d\psi\w [(F'L-G'K)\mu^1 + (G'F-F'G)\mu^2]\ \;,\\
       d\eta^5 = d\mu^2& = \tfrac{1}{\Delta}d\psi\w [(K'\,L-L'K)\mu^1 + (L'F-K'G)\mu^2] \;.
    \end{aligned}
    \]
    which become \Eq{fullStructure} on substituting \Eq{eta-define}. Finally, $d\eta^3 = d^2\psi = 0$, as desired.
\end{proof}

\begin{lemma}[Efficient EDS Form of Realizability] \label{lem:SphereEDS}
  In the coframe $\{\omega,\eta\}$ on $M_U$, the EDS $\cI^0$ is generated by
    \begin{equation}\label{eq:BetaOmega}
    \begin{aligned}
        \beta^1&=\eta^4 \w \eta^3 + \cB\, \omega^1\w \omega^2 \;,\\
        \beta^2&=\eta^5 \w \eta^3 + \cB\, \omega^3\w \eta^3 \;.
    \end{aligned}
    \end{equation}
    with exterior derivatives modulo $\beta^1,\beta^2$ given by  
    \begin{equation} \label{eq:betaDiffs}
    \begin{aligned}    
        d\beta^1&\equiv\cB_2\, \eta^5\w \omega^1\w \omega^2+\cB\, \omega^1\w \omega^3\w \eta^2 - \cB\, \omega^2\w \omega^3\w \eta^1 \mod \{\beta^1,\beta^2\}\; ,\\
        d\beta^2&\equiv\cB\cB_1\, \,\Omega+\cB\,\omega^1\w \eta^1\w d\psi + \cB\,\omega^2\w \eta^2 \w \eta^3 \mod \{\beta^1,\beta^2\}\; ,
    \end{aligned}
    \end{equation}
    where $\cB_1, \cB_2, \cB_3$ are the covariant derivatives of $\cB$ defined through 
    \[ d\cB = \cB_1 \mu^1 + \cB_2 \mu^2 + \cB_3 d\psi  = \cB_3 \eta^3 + \cB_1 \eta^4 + \cB_2 \eta^5. \]
\end{lemma}
\begin{proof}
    Using \Eq{muDef} in \Eq{EDS_I} then gives
    \begin{equation} \label{eq:BetaForms}
    \begin{aligned}
        \beta^1& = \mu^1\w  d\psi+\cB\,\iota_{\bm{b}}\Omega \;,\\
        \beta^2& = \mu^2\w d\psi+\cB\, \bm{b}^\flat \w d\psi \;.
    \end{aligned}
    \end{equation}
    Using the properties of the $\omega^i$ in \Eq{omegaProp} guaranteed by \Lem{fExists}, we have $\omega^3 = \bm{b}^\flat$, $\omega^1 \w \omega^2 = \iota_{\bm{b}} \Omega$, and $\Omega = \omega^1\w\omega^2\w\omega^3$. Using these relations between $\bm{b}$ and $\omega^i$, together with the definition $\eta = (\omega^3_1,\omega^3_2, d\psi, \mu^1,\mu^2)$ we obtain \Eq{BetaOmega}.

    The exterior derivatives of $\beta^1, \beta^2$ can be simplified using the structure equations in \Lem{SphereBundle}. A calculation gives
\begin{align*}
     d\beta^1 &= \cB_2\, \eta^5\w \omega^1\w \omega^2 +\cB\,\omega^1\w\omega^3\w\eta^2 - \cB\,\omega^2\w\omega^3\w\eta^1+ \frac{\cB_1}{\cB}\, \eta^4\w\beta^1 + \frac{\cB_3}{\cB}\, \eta^3\w\beta^1 \;, \\
    d\beta^2 &= \cB\cB_1\, \Omega + \cB\, \omega^1\w\eta^1 \w \eta^3 + \cB\, \omega^2 \w \eta^2 \w \eta^3 - \cB_1\, \beta^1\w\omega^3 - \cB_2\, \beta^2\w\omega^3 \;.
\end{align*}
The result follows by taking the above expression modulo $\beta^1,\beta^2$.
    
\end{proof}

\section{Realizability as a Pfaffian system}\label{sec:ProlongEDS}
The EDS $\cI^0$ of \Lem{SphereEDS} is not generated by 1-forms so the methods of \Sec{Pfaffian_Defined} and onward are not immediately applicable; however, one can always create a new EDS $\cI$ generated by 1-forms on a larger manifold so that integral manifolds of $\cI^0$ are contained in integral manifolds of $\cI$ via ``linearization". This is called \textit{prolongation} (see \App{CK-forms-ver} for further details in the language of differential forms and see \App{PDE_to_Pfaffian} in terms of the presentation from \Sec{Pfaffian_Defined}). In our case,  although \Sec{Cartan-Kahler} does not apply directly to $\cI^0$, it is still possible to determine involutivity of $\cI^0$ and it turns out that the EDS $\cI^0$ is not involutive\footnote
{We do not prove that in the current paper.};
thus, it must be prolonged. An important fact is that a prolongation of any involutive EDS is again involutive \cite{EDS-Book}, so one can always choose to work with a prolonged system from the start without loss of generality (but at the price of working on a higher dimensional space) since prolongation may be a necessary step anyway. As such, we will work with the prolonged system to determine involutivity which then yields a proof of \Th{FTSD}. 

First, we compute the integral elements of $\cI^0$ and then use them to construct a prolonged EDS, $\cI$, which will be shown to be on some subset $M_U^{(1)}\subset M^{(1)} = M\times \bR^8$.

Recall that solutions to \Eq{BoozerPDE} can be understood as graphs of $(\bm{B},\Phi)$ so that $(\bm{x},\bm{B}(\bm{x}),\Phi(\bm{x}))$ are the integral manifolds of $\cI^0$, which must be submanifolds of $M$. In a coordinate free form, this implies that the five $1$-forms
\begin{equation}\label{eq:etaDefine}
    \eta = (\omega^3_1,\omega^3_2,d\psi,\mu^1,\mu^2) \;,\\
\end{equation}
in \Lem{SphereBundle} must become linearly dependent on the $\omega^i$ on the graph of any solution $(\bm{B},\Phi)$ of \Eq{BoozerPDE}. This linear dependence defines the possible 3-dimensional integral elements of $\cI^0$ and indeed, we can instead work directly with the 1-forms that are dual to the integral elements of $\cI^0$. This requires working on a slightly larger space where the additional dimensions represent the coefficients of the $\eta$'s dependence on the $\omega^i$'s that are not determined by the vanishing of the 2-forms and 3-forms $\beta^1,\beta^2,d\beta^1,d\beta^2$ when evaluated on integral elements. 

\begin{lemma}[Pfaffian EDS]\label{lem:prolonged-EDS}
The 3-dimensional integral manifolds of the exterior differential system $\cI^0$ are in one-to-one correspondence with the integral manifolds of a linear Pfaffian system $\cI$ on the open subset $M_U^{(1)}\subset M^{(1)}=M\times \bR^8$ generated by the 1-forms 
\begin{equation}\label{eq:alphaDefine}
    \alpha^a = \eta^a - \sum_{i=1}^3 C_{i}^a\, \omega^i \;, \quad a = 1,\ldots, 5 \;,
\end{equation}
where the coordinates on $\bR^8$ are given by 
\begin{equation}\label{eq:REightCoord}
    p^\epsilon = (c_{12},c_{21},\kappa_1,\kappa_2,\kappa_3,\kappa_4,\psi_1, \psi_2) \;,
\end{equation}
and the matrix $C$ is given by 
\begin{equation}\label{eq:Cdefine}
C = \begin{pmatrix}
    \kappa_1 - \tfrac12 \cB_2 & c_{12} & \kappa_2 \psi_1 + \cB_1 \bar{\psi}_2 \\
    c_{21} & - \kappa_1 - \tfrac12 \cB_2 & \kappa_2 \psi_2 - \cB_1 \bar{\psi}_1  \\
    \psi_1 & \psi_2 & 0 \\
    \kappa_3 \psi_1 - \cB \bar{\psi}_2 & \kappa_3 \psi_2 + \cB\bar{\psi_1} & 0 \\
    \kappa_4 \psi_1 & \kappa_4 \psi_2 & -\cB 
\end{pmatrix},\qquad \bar{\psi}_i \equiv \frac{\psi_i}{\psi_1^2 + \psi_2^2}.
\end{equation}
Moreover, the existence of integral manifolds of $\cI^0$ is equivalent to the existence question of integral manifolds for $\cI$. 
\end{lemma}
\begin{proof}
Since $M$ is an 8-dimensional manifold, then the candidate 3-dimensional integral elements of $\cI^0$ are dual to (annihilated by) the five 1-forms for the space \Eq{etaDefine},
\begin{equation*}
\begin{aligned}
    \alpha^1&=\omega^3_1-C^1_1\omega^1-C^1_2\omega^2-C^1_3\omega^3,\\
    \alpha^2&=\omega^3_2-C^2_1\omega^1-C^2_2\omega^2-C^2_3\omega^3,\\
    \alpha^3&=d\psi-C^3_1\omega^1-C^3_2\omega^2,\\
    \alpha^4&=\mu^1-C^4_1\omega^1-C^4_2\omega^2-C^4_3\omega^3,\\
    \alpha^5&=\mu^2-C^5_1\omega^1-C^5_2\omega^2-C^5_3\omega^3.
\end{aligned}
\end{equation*}
for some matrix of coefficients $C_i^a$.
Notice that $C^3_3=0$ has already been determined by our assumption that $\omega^3=\bm{b}^\flat$ (see \Eq{omegaProp}) and that any solution $\bm{B}$ will be tangent to level surfaces of $\psi$. 
Genuine integral elements of $\cI^0$ will be determined by the 2-forms $\beta^1$ and $\beta^2$ and their exterior derivatives. As such, in order for the $\alpha^i$'s to represent 3-dimensional integral elements for $\cI^0$ then we need 
\[
    \begin{matrix}
    \beta^1&\equiv 0\\
    \beta^2&\equiv 0\\
    d\beta^1&\equiv 0\\
    d\beta^2&\equiv 0
    \end{matrix} \quad \mod\{\alpha^i\} \;.
\]
To do so, we solve for the forms $\{\omega^3_1,\omega^3_2,\mu^1,\mu^2,d\psi\}$ in terms of the $\alpha^i$ so we can then express $\beta^1$ and $\beta^2$ and their exterior derivatives in terms of $\alpha^i$'s as well.
So for $\beta^1$ and $\beta^2$ this yields

\[
    \begin{aligned}
    \beta^1&=(\alpha^5+C^5_1\omega^1+C^5_2\omega^2+C^5_3\omega^3)\wedge(\alpha^3+C^3_1\omega^1+C^3_2\omega^2)+\cB\,\omega^1\wedge\omega^2,\\
    \beta^2&=(\alpha^4+C^4_1\omega^1+C^4_2\omega^2+C^4_3\omega^3+\cB\,\omega^3)\wedge(\alpha^3+C^3_1\omega^1+C^3_2\omega^2).
    \end{aligned}
\]
Since we need to find the possible 3-planes that will annihilate the $\beta^a$'s and their exterior derivatives, it is only necessary to expand the above wedge products modulo the $\alpha^i$'s. This is because the $\alpha^i$'s themselves will annihilate the possible 3-dimensional integral elements. Doing so, we have
\begin{equation}\label{eq:C-eqs}
\begin{aligned}
\beta^1&\equiv (C^4_1C^3_2-C^4_2C^3_1+\cB)\,\omega^1\wedge\omega^2+C^3_1C^4_3\,\omega^3\wedge\omega^1-C^3_2C^4_3\omega^2\wedge\omega^3 \mod \{\alpha^i\},\\
\beta^2&\equiv (C^5_1C^3_2-C^5_2C^3_1)\,\omega^1\wedge\omega^2+(C^5_3+\cB)\,C^3_1\,\omega^3\wedge\omega^1-(C^5_3+\cB)\,C^3_2\,\omega^2\wedge\omega^3\mod \{\alpha^i\},\\
d\beta^1&\equiv (\cB_2C^5_3-\cB(C^1_1+C^2_2))\,\Omega\mod\{\alpha^i,\beta^a\},\\
d\beta^2&\equiv (\cB_1\cB+\cB\,(C^3_1C^2_3-C^3_2C^1_3))\,\Omega\mod\{\alpha^i,\beta^a\}.
\end{aligned}
\end{equation}
We now have an algebraic system of equations that determine the necessary relations on the coefficients of the $\omega^i$'s in the $\alpha^i$'s so that they represent 3-dimensional integral elements.

The system \Eq{C-eqs} is underdetermined for the $C^a_i$, so the remaining $C^a_i$ may be parameterized by the 8 variables \Eq{REightCoord}.
Then the solution matrix of $C^a_i$'s to \Eq{C-eqs} is given by \Eq{Cdefine}.

In particular, \cite[Eq. 105, page 128]{EDS-Book}, rigorously demonstrates the claim that we have a one-to-one correspondence between the integral manifolds of the two EDS $\cI^0$ and $\cI$.  

\end{proof}

\section{The realizability theorem}\label{sec:Proof}
\subsection{Determining the tableau and apparent torsion}

We now show that the EDS $\cI$ passes Cartan's test, and is therefore involutive. Note that $\cI$ is a linear Pfaffian system differentially generated by $\{\alpha^a\}$. In consequence, the structure equations for $\alpha^a$ will be of the form 
\begin{equation}\label{eq:AlphaStructure}
\begin{aligned}
    d\alpha^a &= \sum_{i=1}^3 \Pi^a_i\w\omega^i + \sum_{j,k=1}^3 T^a_{jk} \omega^j\w\omega^k \mod \{\alpha^i\} \;, \\
    \Pi^a_i &= \sum_{\epsilon = 1}^8 A^a_{\epsilon i} \pi^{\epsilon} \;, \quad
    \pi^\epsilon \equiv dp^\epsilon \;.
\end{aligned}
\end{equation}
The first term in \Eq{AlphaStructure} gives the associated tableau while the second term is the apparent torsion. To use the Cartan-K\"ahler theorem, it must be shown that the tableau passes Cartan's test and that the apparent torsion is absorbable. 

We begin by computing the tableau. Differentiation of \Eq{alphaDefine} gives 
\begin{equation}\label{eq:TableauStep1}
    d\alpha^a = d\eta^a - \sum_{i=1}^3 \left(dC^a_i\w \omega^i - C^a_i d\omega^i \right) \;.
\end{equation}
We next look at the terms in \Eq{TableauStep1} to determine which contribute to the tableau, and which to the apparent torsion.

\begin{enumerate}[label={(\alph*)}]
    \item Consider first the $d\eta^a$ term in \Eq{TableauStep1}. Using the structure equations \Eq{fullStructure}, $d\eta^a$ has no dependence on $\pi^\epsilon$ or $\omega^i$. By \Eq{alphaDefine}, $\eta^a = \alpha^a + \sum_i C^a_{i} \omega^i$ so all the terms of the form $\eta^a \w \eta^b$ in the structure contribute only to the apparent torsion modulo $\{\alpha^i\}$.

    \item Next, consider the $\sum_i C^a_i d\omega^i$ term in \Eq{TableauStep1}. The structure equations for $\omega^i$ given in \Eq{omegaStructure} do not depend on $\pi^\epsilon$.
    Again, since $\eta^a = \alpha^a + \sum_{i=1}^3 C^a_{i} \omega^i$, all terms of the form $\eta^a \w \omega^i$ in $d\omega^i$ contribute only to the apparent torsion modulo $\{\alpha^i\}$. 

    \item Finally, consider the $\sum_i dC^a_i \w \omega^i$ term in \Eq{TableauStep1}. These contribute to both the tableau and the apparent torsion. Indeed, $C^a_i$ depends on the variables $p^\epsilon$ as well as $(\psi,\theta,\phi)$ through the functions $\cB, \cB_1, \cB_2$. Observe that, if $h = h(\psi,\theta,\phi)$ is any function of these coordinates and we denote 
    \[ dh = h_1 \mu^1 + h_2 \mu^2 + h_3 d\psi = h_1 \eta^4 + h_2 \eta^5 + h_3 \eta^3, \]
    then modulo $\{\alpha^i\}$ we must have 
    \[ dh = \left[( h_1 \kappa_3 + h_2 \kappa_4 + h_3)\psi_1 - \cB h_1 \bar{\psi}_2\right]\, \omega^1 + \left[( h_1 \kappa_3 + h_2 \kappa_4 + h_3)\psi_2 + \cB h_1 \bar{\psi}_1\right]\, \omega^2 - \cB h_2\, \omega^3.  \]
    Hence, the differentials with respect to $\cB, \cB_1,\cB_2$ will contribute only to the apparent torsion terms. 

    However, all differentials of the remaining variables $p^\epsilon$ contribute to the tableau (and in fact are the only terms that do). Denote by $d_\epsilon C^a_i$ the differential only in terms of the $p^\epsilon$ variables. Then, recalling $\pi^\epsilon \equiv dp^\epsilon$ we have that 
    \begin{equation} \label{eq:Tableau}
        d_\epsilon C^a_i = {\begin{pmatrix}
        \pi^3   & \pi^1     & \cB_1 d\bar{\psi}_2 + \psi_1 \pi^4 + \kappa_2 \pi^7 \\
        \pi^2   & -\pi^3    & -\cB_1 d\bar{\psi}_1 + \psi_2 \pi^4 + \kappa_2 \pi^8 \\
        \pi^7 & \pi^8 & 0 \\
        -\cB d\bar{\psi}_2 + \psi_1 \pi^5+\kappa_3 \pi^7 & \cB d\bar{\psi}_1 + \psi_2 \pi^5 + \kappa_3 \pi^8 & 0 \\
        \psi_1 \pi^6+\kappa^4 \pi^7 & \psi_2 \pi^6+\kappa_4 \pi^8 & 0 
        \end{pmatrix}}
    \end{equation} 
    where 
    \[ d\bar{\psi}_1 = -(\bar{\psi}_1^2 - \bar{\psi}_2^2)\, \pi^7 - 2\bar{\psi}_1 \bar{\psi}_2\, \pi^8,\qquad d\bar{\psi}_2 = -2\bar{\psi}_1\bar{\psi}_2\, \pi^7 + (\bar{\psi}_1^2 - \bar{\psi}_2^2)\, \pi^8.  \]
\end{enumerate}

In conclusion, the tableau is $A^a_{\epsilon i} \pi^\epsilon = -d_{\epsilon}C^a_{i}$ with $d_{\epsilon} C^a_{i}$ given by \Eq{Tableau}. The apparent torsion term is computed as $$ \sum_{i,j=1}^3 T^a_{ij} \omega^i\w\omega^j = d\eta - \sum_{i=1}^3 C^a_i d\omega^i - \sum_{i,j=1}^3 \mathcal{F}_{ij} \omega^i\w\omega^j \, \mod \{\alpha^i\}$$
with $\mathcal{F}_{ij}$ due to the exterior derivative of the functions $\cB,\cB_1,\cB_2$ in $C^a_i$.  

\subsection{The apparent torsion is absorbable}\label{sec:Absorb_Torsion}

We proceed by showing that the apparent torsion is absorbable. To do this, it must be shown that there exists a new coframing $\tilde \pi$ with 
\[
    \pi^{\epsilon} = \tilde{\pi}^{\epsilon} + \sum_{i=1}^3 R^{\epsilon}_i \omega^i
\]
such that the $\alpha$ structure equations \Eq{AlphaStructure} become
\[
    d\alpha^a = \sum_{\epsilon =1}^8 \sum_{i=1}^3 A^a_{\epsilon i} \tilde{\pi}^\epsilon \w \omega^i \mod \{\alpha^i\}.
\]
In other words, the apparent torsion vanishes in a particular choice of coframe. In terms of $R^{\epsilon}_i, A^a_{\epsilon i}, T^a_{ij}$ it is explicitly computed that
\[
    d\alpha^a = \sum_{\epsilon = 1}^8 \sum_{i=1}^3 A^{a}_{\epsilon i} \tilde{\pi}^\epsilon\w \omega^i + \sum_{i,j=1}^3\left(T^a_{ij} - \frac{1}{2}\sum_{\epsilon = 1}^8\left( A^a_{\epsilon i} R^\epsilon_j - A^a_{\epsilon j} R^\epsilon_i \right) \right) \omega^i\w\omega^j\;,
\] 
so that the torsion is absorbable if and only if there exists $R^\epsilon_i$ such that 
\begin{equation}
    \tfrac{1}{2}\sum_{\epsilon =1}^8 \left( A^a_{\epsilon i} R^\epsilon_j - A^a_{\epsilon j} R^\epsilon_i \right) = T^a_{ij}.
\end{equation}

Now, for the particular tableau for our prolonged system, consider the new frame $\tilde{\pi}^\epsilon$ given by taking $R$ as
\begin{equation}\label{eq:RSol}
    \begin{aligned}
        \tfrac12 R &= \begin{pmatrix}
        0  & 0         &  -T^1_{23} + \bar{\psi}_1\bar{\psi_2}\left(T^1_{13}+T^2_{23}\right) -  \left(\kappa_2 - 2 \cB_1 \bar{\psi}_1\bar{\psi}_2\right) T^3_{12} \\
        0           & 0  &  -T^2_{13} + \bar{\psi}_1\bar{\psi_2}\left(T^1_{13}+T^2_{23}\right) + \left(\kappa_2 + 2 \cB_1 \bar{\psi}_1\bar{\psi}_2\right) T^3_{12}  \\
        -T^2_{12}    & -T^1_{12}  & -T^1_{13} + \bar{\psi}_1^2 \left(T^1_{13}+T^2_{23}\right) +\cB_1\left(\bar{\psi}_1^2 - \bar{\psi}_2^2 \right) T^3_{12}   \\
        \bar{\psi}_1\left(T^1_{13}+T^2_{23}\right) & \bar{\psi}_2 \left(T^1_{13}+T^2_{23}\right) & 0 \\
        \bar{\psi}_2 \left(T^4_{12}-\kappa_3 T^3_{12}\right) & -\bar{\psi}_1 \left(T^4_{12}-\kappa_3 T^3_{12}\right) & -\tfrac12 R^5_3 \\
        \bar{\psi}_2 \left(T^5_{12} - \kappa_4 T^3_{12}\right)  & -\bar{\psi}_1 \left(T^5_{12} - \kappa_4 T^3_{12}\right) & -\tfrac12 R^6_3 \\
        0   &  -T^3_{12} &  -T^3_{13} \\
        0 & 0 &  -T^3_{23}
        \end{pmatrix}
        \\
        \tfrac12 R^5_3 &= \left(\bar{\psi}_1 T^4_{13} - \left(\kappa_3 \bar{\psi}_1 + \frac{\bar{\psi}_2}{\psi_1^2 + \psi_2^2} \cB \right)T^3_{13}\right) + \left(\bar{\psi}_2 T^4_{23} - \left(\kappa_3 \bar{\psi}_2 - \frac{\bar{\psi}_1}{\psi_1^2 + \psi_2^2} \cB \right) T^3_{23}\right)\\
        \tfrac12 R^6_3 &=  \bar{\psi}_1 \left( T^5_{13} - \kappa_4 T^{3}_{13} \right) + \bar{\psi}_2 \left( T^5_{23} - \kappa_4 T^3_{23}\right).
    \end{aligned}
\end{equation}
A rather lengthy calculation shows that for the structure equations \Eq{AlphaStructure}
all of the terms in the new apparent torsion vanish except for the possibly nonzero terms 
\[ 
\begin{aligned}
    \tilde{T}^4_{13} =-\tilde{T}^4_{31} &= -\bar{\psi}_2 [T]^4 \;, &
    \tilde{T}^4_{23} = -\tilde{T}^4_{32} &= \bar{\psi}_1 [T]^4 \;,\\
    \tilde{T}^5_{13} = -\tilde{T}^5_{31} &= -\bar{\psi}_2 [T]^5 \;, &
    \tilde{T}^5_{23} = -\tilde{T}^5_{32} &= \bar{\psi}_1 [T]^5 \;,
\end{aligned}
\]
with 
\[
\begin{aligned}
   \relax [T]^4 &= \left(\psi_1 T^4_{23} - \left(\kappa_3 \psi_1 - \cB \bar{\psi}_2\right) T^3_{23} \right) -\left(\psi_2 T^4_{13} - \left(\kappa_3 \psi_2 + \cB \bar{\psi}_1\right) T^3_{13}\right)   \\
    [T]^5 &=  \psi_1 \left( T^5_{23} - \kappa_4 T^3_{23} \right) - \psi_2\left(T^5_{13} - \kappa_4 T^3_{13}\right). 
\end{aligned}
\]
Hence, in order to show that the torsion is absorbable, it remains only to show that $[T]^4 = 0$ and $[T]^5 = 0$. The following lemma reveals that $[T]^4$ and $[T]^5$ arise from $d\beta^1,d\beta^2$. 

\begin{lemma}\label{lem:TsVanish}
    Let $T^a \equiv T^a_{ij} \omega^i\w\omega^j$ be the apparent torsion $2$-forms. It holds that 
    \begin{align*}
        [T]^4 \, \Omega &= T^4\w\eta^3 - T^3\w \eta^4  = d\beta^1 \mod \{\alpha^i\} \\
        [T]^5 \, \Omega &= (T^5 - \kappa_4 T^3) \w \eta^3  = d\beta^2 \mod \{\alpha^i\}.
    \end{align*}
    Hence $[T]^4 = [T]^5 = 0$. 
\end{lemma}
\begin{proof}
    Recalling that 
    \begin{align*}
        (C^3_1, C^3_2, C^3_3) &= (\psi_1,\, \psi_2,\, 0), \\
        (C^4_1, C^4_2, C^4_3) &= (\kappa_3 \psi_1 -  \cB\bar{\psi}_2,\, \kappa_3 \psi_2 + \cB\bar{\psi}_1,\, 0) \\
        (C^5_1, C^5_2, C^5_3) &= (\kappa_4 \psi_1,\, \kappa_4 \psi_2,\, -\cB),
    \end{align*}
    a computation shows
    \begin{align*}
        T^4\w\eta^3 - T^3\w\eta^4 &= (- T^4_{13} C^3_2 + T^4_{23} C^3_1) - ( - T^3_{13} C^4_2 + T^3_{23} C^4_1)\, \Omega \mod \{\alpha^i\} \\
        &= [T]^4\, \Omega \mod \{\alpha^i\} \\
        (T^5-\kappa_4 T^3)\w\eta^3 &= -\left( (T^5_{13} - \kappa_4 T^3_{13}) C^3_2 + (T^5_{23} - \kappa_4 T^3_{23}) C^3_1 \right) \Omega \mod \{\alpha^i\} \\
        &= [T]^5 \, \Omega \mod \{\alpha^i\}.
    \end{align*}
    To show equivalence to $d\beta^i$ modulo $\{\alpha^i\}$, first we compute $T^3, T^4, T^5$. As outlined above, the torsion terms are computed as 
    \[ T^a = d\eta^a - \sum_{i=1}^3 C^a_i d\omega^i - \mathcal{F}^a, \]
    where $\mathcal{F}^a$ are the terms of $dC^a_i\w \omega^i$ that come from exterior derivatives of the functions $\cB,\cB_1,\cB_2$. Explicitly, we have 
    \begin{align*}
        T^3 &= d\eta^3 -\psi_1 d\omega^1 - \psi_2 d\omega^2 \mod \{\alpha^i\}\\
        T^4 &= d\eta^4 - (\kappa_3 \psi_1 -  \cB\bar{\psi}_2) d\omega^1 - (\kappa_3 \psi_2 + \cB\bar{\psi}_1) d\omega^2 + \bar{\psi}_2 d\cB \w \omega^1 - \bar{\psi}_1 d\cB \w \omega^2 \mod \{\alpha^i\} \\
        T^5 &= d\eta^5 - \kappa_4(\psi_1 d\omega^1 + \psi_2 d\omega^2) + \cB d\omega^3+ d\cB \w \omega^3 \mod \{\alpha^i\}.
    \end{align*}
    Now, from the structure equations \Eq{fullStructure} it holds that $d\eta^3 = 0$ and $d\eta^4\w\eta^3 = 0 = d\eta^5 \w\eta^3$. Therefore,
    \begin{align*}
        T^3\w\eta^4 &= -\psi_1 (\kappa_3 \psi_1 -  \cB\bar{\psi}_2) d\omega^1\w\omega^1 - \psi_1 (\kappa_3 \psi_2 + \cB\bar{\psi}_1) d\omega^1\w\omega^2 \\
        &\qquad - \psi_2 (\kappa_3 \psi_1 -  \cB\bar{\psi}_2) d\omega^2\w\omega^1 - \psi_2 (\kappa_3 \psi_2 + \cB\bar{\psi}_1)d\omega^2\w\omega^2 \mod \{\alpha^i\} \\
        T^4\w\eta^3 &= - \psi_1 (\kappa_3 \psi_1 -  \cB\bar{\psi}_2) d\omega^1\w\omega^1 - \psi_2 (\kappa_3 \psi_1 -  \cB\bar{\psi}_2) d\omega^1\w\omega^2 \\
            &\qquad -\psi_1 (\kappa_3 \psi_2 + \cB\bar{\psi}_1) d\omega^2\w\omega^1 - \psi_2(\kappa_3 \psi_2 + \cB\bar{\psi}_1) d\omega^2\w\omega^2 \\
            &\qquad + d\cB\w\omega^1\w\omega^2 \mod \{\alpha^i\}.
    \end{align*}
    Finally,
    \begin{align*}
        T^4\w\eta^3 - T^3\w\eta^4 &= \cB d\omega^1\w\omega^2 - \cB d\omega^2\w\omega^1 + d\cB\w\omega^1\w\omega^2 = d\beta^1 \mod \{\alpha^i\}
    \end{align*}
    as desired.

    For the term $[T]^5$ observe that $T^5 - \kappa_4 T^3 = d\eta^5 +\cB d\omega^3 + d\cB\w\omega^3 \mod \{\alpha^i\}$. Then
    \begin{align*}
        (T^5-\kappa_4 T^3)\w\eta^3 &=  \cB d\omega^3\w\eta^3 + d\cB\w\omega^3 \w\eta^3 = d\beta^2 \mod \{\alpha^i\}
    \end{align*}
    and the result follows.
\end{proof}

By \Lem{TsVanish} we can conclude that the choice of $R$ given in \Eq{RSol} completely absorbs all apparent torsion.

\subsection{The tableau is involutive}\label{sec:Involutive}

In this section it is demonstrated the tableau is involutive, recall \Step{Cartan} in \Sec{CKAlg}. The involutivity of the tableau together with the absorption of all apparent torsion will prove \Th{FTSD}. The tableau will be involutive provided the reduced Cartan characters $s_1, s_2, s_3$ satisfy Cartan's test $s_1 + 2 s_2 + 3 s_3 = r$ where $r$ is the \textit{degree of indeterminacy} \cite{Olver95}.

We begin by computing the Cartan characters of the tableau. The tableau written explicitly in the coframe $(\omega^i, \eta^a, \pi^\epsilon)$ is given by \Eq{Tableau}. Provided the coframe is \textit{generic}, the Cartan character $s_k$ is the number of 1-forms in the $k^{th}$ column of the tableau that are independent of 1-forms in the first $k-1$ columns. Unfortunately, the chosen coframe $(\omega^i, \eta^a,\pi^\epsilon)$ turns out not to be generic. For a generic coframe we take 
\[
    \omega^1 = \tilde{\omega}^1,\qquad 
    \omega^2 = \tilde{\omega}^2,\qquad 
    \omega^3 = \psi_1 \tilde{\omega}^1 + \psi_2 \tilde{\omega}^2 + \tilde{\omega}^3 \;.
\]
Further, assuming without loss of generality that $\psi_1 \neq 0$ on integral elements, the tableau can be brought into normal form by introducing the coordinates $\tp^\epsilon = \tp^\epsilon(c_{12},c_{21}, \kappa_1,\kappa_2,\kappa_3,\kappa_4,\psi_1,\psi_2)$ given explicitly by 
\[
    \begin{aligned}
        \tp^1 &= \kappa_1 + \psi_1 (\kappa_2 \psi_1 + \cB_1 \bar{\psi}_2) \;,
              & \qquad  \tp^5 &= \psi_1 (\kappa_4 - \cB) \;,\\
        \tp^2 &= c_{21} + \psi_1(\kappa_2 \psi_2 - \cB_1 \bar{\psi}_1)   \;,
              & \tp^6 &= c_{12} + \psi_2(\kappa_2 \psi_1 + \cB_1 \bar{\psi}_2) \;,\\
        \tp^3 &= \psi_1     \;,                                     
              & \tp^7 &= -\kappa_1 + \psi_2(\kappa_2 \psi_2 - \cB_1 \bar{\psi}_1) \;, \\
        \tp^4 &= \kappa_3 \psi_1 - \cB \bar{\psi}_2    \;,    
              &  \tp^8 &= \psi_2   \;.
    \end{aligned}
\]
In these new coordinates we have $\alpha^a = \eta^a - \sum_{i=1}^3 \tilde{C}_i^a \tilde{\omega}^i$ where $\tilde{C}$ is given by 
\[
    \tilde{C} = \begin{pmatrix}
        -\tfrac12 \cB_2 + \tp^1    & \tp^6 &\frac{\tp^3 (\tp^1 + \tp^7) + \cB_1 \tp^8}{(\tp^3)^2 + (\tp^8)^2} \\
        \tp^2   & -\tfrac12 \cB_2 +  \tp^7     &  \frac{-\cB_1\tp^3 + (\tp^1 + \tp^7) \tp^8}{(\tp^3)^2 + (\tp^8)^2} \\
        \tp^3   & \tp^8     & 0 \\
        \tp^4   & \frac{\cB + \tp^4\tp^8}{\tp^3} & 0\\
        \tp^5   & \frac{\tp^5 \tp^8}{\tp^3} & -\cB
    \end{pmatrix}.
\]
Then, setting $\tilde{\pi}^{\epsilon} = d\tp^\epsilon$ the tableau in the coframing $(\tilde{\omega}^i, \eta^i, \tilde{\pi}^i)$ is
\[
    \tA^a_{\epsilon i} \tilde{\pi}^\epsilon = -\begin{pmatrix}
        \tilde{\pi}^1 & \tilde{\pi}^6 & \bar{\psi}_1 (\tilde{\pi}^1  + \tilde{\pi}^7) + \tA^1_{33} \tilde{\pi}^3 + \tA^1_{83} \tilde{\pi}^8 \\
        \tilde{\pi}^2 & \tilde{\pi}^7 & \bar{\psi}_2 (\tilde{\pi}^2 + \tilde{\pi}^7) + \tA^2_{33} \tilde{\pi}^3 + \tA^2_{83} \tilde{\pi}^8 \\
        \tilde{\pi}^3 & \tilde{\pi}^8 & 0 \\
        \tilde{\pi}^4 & \frac{\tp^4}{\tp^3} \tilde{\pi}^8 + \frac{\tp^8}{\tp^3} \tilde{\pi}^4 - \frac{\cB+ \tp^4\tp^8}{(\tp^3)^2} \tilde{\pi}^3 & 0 \\
        \tilde{\pi}^5 & \frac{\tp^5}{\tp^3} \tilde{\pi}^8 + \frac{\tp^8}{\tp^3} \tilde{\pi}^5 - \frac{ \tp^5\tp^8}{(\tp^3)^2} \tilde{\pi}^3 & 0
    \end{pmatrix}
\] 
where
\[
    \tA^1_{33} = -\tA^2_{83} = \frac{\partial}{\partial \tp^3} \left(\frac{\tp^3 (\tp^1 + \tp^7) + \cB_1 \tp^8}{(\tp^3)^2 + (\tp^8)^2}\right), \qquad
    \tA^1_{83} = \tA^2_{33} = \frac{\partial}{\partial \tp^8}  \left(\frac{\tp^3 (\tp^1 + \tp^7) + \cB_1 \tp^8}{(\tp^3)^2 + (\tp^8)^2}\right).
\]
Using \Step{Cartan}, it is now immediate that the Cartan characters are $s_1 = 5, s_2 = 3, s_3 = 0$. 

Finally, we compute the degree of indeterminacy $r$. The degree of indeterminacy is the amount of freedom in the solution $R$ to the absorption equation 
$$ T^a_{ij} = \frac{1}{2}\sum_{\epsilon =1}^8 \left( A^a_{\epsilon i} R^\epsilon_j - A^a_{\epsilon j} R^\epsilon_i \right). $$
In other words, it is the dimension of the kernel of the operator $R \mapsto \tfrac{1}{2} \sum_{\epsilon=1}^8 \left( A^a_{\epsilon i} R^\epsilon_j - A^a_{\epsilon j} R^\epsilon_i \right)$. A direct computation shows the degree of indeterminacy is $r = 11$. Indeed, in the original coframing $(\omega^i,\eta^a,\pi^\epsilon)$ with the tableau given in \Eq{Tableau}, the elements $R_{\rm ker}$ in the kernel of the absorption equation are of the form 
\begin{equation}
    R_{\rm ker} = \begin{pmatrix}
        \sigma_1 & \sigma_4 & R^1_3 \\
        \sigma_2 & -\sigma_3 & R^2_3 \\
        \sigma_3 & \sigma_1 & R^3_3 \\
        R^4_1 & \psi_1\, \sigma_5  & \sigma_{11} \\
        R^5_2 & \psi_2\, \sigma_6 & 0 \\
        \psi_1\, \sigma_{7} & \psi_2\, \sigma_{7} & 0 \\
        \psi_1 \psi_2\, \sigma_8 & \psi_1 \psi_2\, \sigma_9 & 0 \\
        \psi_1\psi_2\, \sigma_9 & \psi_1 \psi_2\, \sigma_{10} & 0
    \end{pmatrix},\quad \sigma_1,\dots,\sigma_{11} \in \bR,
\end{equation}
where 
\[
\begin{aligned}
    R^4_1 &= -\psi_2
    \left( \sigma_5 + (\kappa_2 -2\cB_1 \bar{\psi}_1 \bar{\psi}_2)\, \sigma_8 + 2 \cB_1 (\bar{\psi}_1^2 - \bar{\psi}_2^2)\, \sigma_9 + (\kappa_2 + 2\cB_1 \bar{\psi}_1 \bar{\psi}_2)\,\sigma_{10}\right) \;, \\
    R^5_2 &= \psi_1 \left( \sigma_6 + \cB (\bar{\psi}_1^2 - \bar{\psi}_2^2)\, \sigma_8 + 4 \cB \bar{\psi}_1 \bar{\psi}_2\, \sigma_9 - \cB (\bar{\psi}_1^2 - \bar{\psi}_2^2)\, \sigma_{10}\right) \;, \\
    R^1_3 &= \psi_1 \left( \psi_1\, \sigma_5 + \psi_2(\kappa_2 -2\cB_1 \bar{\psi}_1 \bar{\psi}_2)\;\sigma_9 + \cB_1 \psi_2 (\bar{\psi}_1^2 - \bar{\psi}_2^2)\;\sigma_{10} \right) \;, \\
    R^2_3 &= -\psi_2 \left(  
        \psi_2\; \sigma_5
        + (\kappa_2 \psi_2 - \cB_1 \bar{\psi}_1)\; \sigma_8
        - (\kappa_2 \psi_1 + 2 \cB_1 \bar{\psi}_2^2 \psi_2)\;\sigma_9
        + \psi_2 (\kappa_2 + 2\cB_1 \bar{\psi}_1 \bar{\psi}_2)\; \sigma_{10}
    \right) \;, \\
    R^3_3 &= -\psi_1 \psi_2 \left( 
        \sigma_5 + \cB_1 (\bar{\psi}_1^2 - \bar{\psi}_2^2)\; \sigma_9
        + (\kappa_2 + 2\cB_1 \bar{\psi}_1 \bar{\psi}_2)\, \sigma_{10}
    \right) \;.
\end{aligned}
\]
It follows that $s_1 + 2 s_2 + 3s_3 = 5 + 2\times 3 + 3 \times 0 = 11 = r$. Hence, the tableau passes Cartan's test and it can be concluded that it is involutive. The involutivity of the tableau together with the fact that the apparent torsion is absorbable guarantee that the local analytic integral manifolds of $\cI$ exist. Thus we have shown that the realizability theorem, \Th{FTSD}, holds.

\section{Conclusions}
Our main result (\Th{FTSD}) says that, locally, a genuine magnetic field exists in Euclidean $\bR^3$ for any analytic choice of $\cB$ in Boozer coordinates. Moreover, our realizability theorem shows that two functions of one variable and three functions of two variables are needed to specify the Cauchy data for the realizability PDE system \Eq{BoozerPDE} in order to guarantee a unique analytic magnetic field with the given field strength and functions $F,G,K,L$. In this sense, \Th{FTSD} provides information on the amount of freedom available in the space of magnetic fields that have a prescribed field strength.

The information about the required Cauchy data can be used to inform optimization efforts for stellarator design. For example, if  an optimization is first performed to ensure a magnetic field matches some prescribed field strength and $F,G,K,L$, then one knows from \Th{FTSD} that there are at most only two functions of one variable and three functions of two variables of freedom left to optimize for other useful stellarator properties. If optimizing for another property requires more freedom than this, or lives outside the space of possible Cauchy data, then a trade-off must be made. 

There are three interesting questions that are raised by this paper and seem worth pursuing:
\begin{enumerate}
    \item[(1)] Do the local solutions extend globally into an entire toroidal plasma volume, such as a neighborhood of a magnetic axis? If not, what restrictions on the Cauchy data provide a global solution?
    \item[(2)] Can the method lead to an appropriate computational algorithm?
    \item[(3)] Do the magnetic fields realizing a particular magnitude $\cB$ satisfy force-balance relations, such as the MHS equation \Eq{MHS}? What is the relationship between the two free single-variable functions in \Th{FTSD} and the usual free functions encountered in the theory of MHS solutions with nested flux surfaces?  What if one also demands that $\bm{B}$ is quasi-symmetric?
\end{enumerate}

For the question of global existence, (1), there is at least one primary obstruction from the sphere bundle construction in \Sec{BoozerEDS}; however, this puts only a minimal restriction the direction of the magnetic field $B$, so it remains possible that magnetic fields can exist on the entire toroidal plasma volume. An investigation of such global results would require understanding how periodicity can be merged with the Cartan-K\"ahler theorem for EDS. Moreover, extra care needs to be taken near the axis because $\nabla \psi \neq 0$ is an assumption throughout the current paper. However, Boozer coordinates can be defined smoothly near the axis \cite{Burby21}.

To develop an algorithm, question (2), one can in principle recover the sequence of two Cauchy problems in a choice of local coordinates for $\bR^3$. To do so requires unpacking the tableau in \Sec{Proof} using some chosen local coordinates. While this is algebraically complex in general, it is worth exploring in future work. Regardless, the guarantee of existence of analytic solutions locally implies that power series methods can be naively applied without worry of convergence, at least locally.  

For question (3), the following result shows that not all magnetic fields that have Boozer coordinates satisfy the force-balance relation \Eq{MHS}.
\begin{lemma}[Categorization of $\bm{B}$ admitting Boozer Coordinates]
   Let $p$ be a function on any domain $Q\subset \bR^3$ such that $\nabla p \neq 0$ almost-everywhere. A non-vanishing, divergence-free magnetic field $\bm{B}$ admits Boozer coordinates $(\psi,\theta,\zeta)$ with $p = p(\psi)$ in the neighborhood of any regular toroidal surface of $p$ if and only if $(\nabla\times B)\cdot \nabla p = 0$ and $B\cdot\nabla p = 0$
\end{lemma}
\begin{proof}
    Suppose that $\bm{B}$ admits Boozer coordinates near every regular toroidal surface of $p$. Choose any regular toroidal surface and let $U$ be the neighborhood of this surface that admits Boozer coordinates. Then, in $U$, using \Eq{BoozerC1},
    \[
        B\cdot \nabla p = \rho(\psi)\left(F(\psi)\partial_\zeta p + G(\psi)\partial_\theta p\right) = 0 \;.
    \]
    Moreover, recalling that $J^\psi = 0$ in Boozer coordinates, implies that 
    $(\nabla\times B)\cdot \nabla p = J^\psi p'(\psi) = 0$.
    Hence, every regular surface of $p$ admits a neighborhood in which $\bm{B}$ behaves as desired. At points where $\nabla p = 0$ the result is automatic. Hence the forward implication holds on the entire domain.

    Now suppose that $(\nabla\times B)\cdot\nabla p = 0 = B\cdot \nabla p$. If $\Omega$ is the Euclidean volume form then, in the language of  \cite{Perrella23} this establishes $(B,\Omega, dp)$ as a \textit{flux system} with an adapted 1-form $B^\flat$. The results from \cite{Perrella23} namely Thm.~I.3 followed by Thm.~V.2 and Prop.~V.3, then establish Boozer coordinates near any regular toroidal surface in $Q$. 
\end{proof}

One may ask instead---or in addition to force-balance---that $\bm{B}$ be quasi-symmetric (QS). Indeed, quasi-symmetric configurations are often presented in terms of $\cB$, and our \Th{FTSD} says any such analytic QS choice of $\cB$ is locally realizable; however, it is possible that the actual flux surface geometries are still highly constrained to known cases. This is at the heart of the paper  \cite{burbyCharacterizationAdmissibleQuasisymmetries2025}

Finally, note that \Th{FTSD} holds for any choice of Riemannian metric $g$. Indeed, the formulation of the realizability problem in terms of the forms $\omega^i, \omega^i_j$ is done for any metric in \App{Convenient_Frame}. The only change in the structure equations \Eq{fullStructure} is to the derivatives $d\omega^i_j$ by the addition of curvature terms and that $\omega^1_2$ now also depends on $\omega^i$. However, the involutivity and absorbability shown in \Sec{Proof} does not depend on these added terms. Thus, the Cartan-K\"ahler theorem shows that \Th{FTSD} holds for any Riemannian metric.

\section*{Acknowledgements}
This work was supported by US Department of Energy Contract DE-FG05-80ET-53088 (JWB) and the Simons Foundation through Grant No. 601972 ``Hidden Symmetries and Fusion Energy" (TK and JDM). ChatGPT (version 5.5) was used to review a draft of Appendix C. This review led to a correction of a small error in the draft of a proof. 

\newpage

\bibliographystyle{plainnat}

\bibliography{cumulative_bib_file.bib}

@PREAMBLE{
 "\providecommand{\noopsort}[1]{}" 
 # "\providecommand{\singleletter}[1]{#1}%" 
}

@book{YangHyperbolicCK,
  title={Involutive hyperbolic differential systems},
  author={Yang, D.},
  year={1987},
  publisher={American Mathematical Society}
}

@book{EDS-Book,
  title={Exterior differential systems},
  author={Bryant, R.L. and Chern, S-S and Gardner, R.B. and Goldschmidt, H.L. and Griffiths, P.A.},
  volume={18},
  year={2013},
  publisher={Springer Science \& Business Media}
}

@book{Ivey16,
   author = {Ivey, T.A. and Landsberg, J.M.},
   title = {Cartan for Beginners: Differential Geometry Via Moving Frames and Exterior Differential Systems},
   publisher = {American Mathematical Society},
   address = {Providence},
   volume = {175},
   edition = {2nd},
   series = {Graduate Studies in Mathematics},
   url = {https://doi.org/10.1090/gsm/175},
   year = {2016}
}

@book{Courant89,
   author = {Courant, R. and Hilbert, D.},
   title = {Methods of Mathematical Physics},
   publisher = {John Wiley \& Sons},
   address = {New York},
   volume = {II},
   series = {Wiley Classics Library},
   year = {1989}
}

@article{Burby21,
   author = {Burby, J.W. and Duignan, N. and Meiss, J.D.},
   title = {Integrability, Normal Forms and Magnetic Axis Coordinates},
   journal = {J.  Math. Phys.},
   volume = {62},
   number = {12},
   pages = {122901},
   url = {https://doi.org/10.1063/5.0049361},
    year = {2021}
}

@article{Burby23,
   author = {Burby, J.W. and Duignan, N. and Meiss, J.D.},
   title = {Minimizing Separatrix Crossings through Isoprominence},
   journal = {Plasma Physics and Controlled Fusion},
   volume = {65},
   number = {4},
   pages = {045004},
   url = {https://doi.org/10.1088/1361-6587/acb968},
   year = {2023}
}

@article{Burby23b,

   author = {Burby, J.W. and MacKay, R.S. and Naik, S.},
   title = {Isodrastic Magnetic Fields for Suppressing Transitions in Guiding-Centre Motion},
   journal = {Nonlinearity},
   volume = {36},
   number = {11},
   pages = {5884},
   url = {https://doi.org/10.1088/1361-6544/acf26a},
   year = {2023}

}

@article{MacKay20a,
   author = {MacKay, R.S.},
   title = {Differential Forms for Plasma Physics},
   journal = {J. Plas. Phys.},
   volume = {86},
   number = {1},
   pages = {925860101},
   url = {https://doi.org/10.1017/S0022377819000928},
   year = {2020}
}

@BOOK{Northrop_1963,
   author       = {T.G. Northrop},
   lccn         = {lc63022462},
   series       = {Interscience tracts on physics and astronomy},
   year         = 1963,
   title        = {The Adiabatic Motion of Charged Particles},
   publisher    = {Interscience Publishers}
}

@article{Kruskal58,
   author = {Kruskal, M. and Kulsrud, R. M.},
   title = {Equilibrium of a Magnetically Confined Plasma in a Toroid},
   journal = {Phys. Fluids},
   volume = {1},
   pages = {265–274},
   url = {https://doi.org/10.1063/1.1705884},
    year = {1958}
}

@ARTICLE{Littlejohn_1981,
   author       = "R. G. Littlejohn",
   title        = "Hamiltonian formulation of guiding center motion",
   year         = "1981",
   journal      = "Phys.\ Fluids",
   volume       = "24",
   pages        = "1730",
   url          = "https://doi.org/10.1063/1.863594"
}

@ARTICLE{Littlejohn_1983,
   author       = "R. G. Littlejohn",
   title 				= "Variational principles of guiding centre motion",
   year         = "1983",
   journal      = "J. Plasma Phys.",
   volume       = "29",
   pages        = "111",
   url          ="https://doi.org/10.1017/S002237780000060X"
}

@book{Hazeltine03,
   author = {Hazeltine, R.D. and Meiss, J.D.},
   title = {Plasma Confinement},
   publisher = {Dover Publications},
   address = {Mineola, NY},
   edition = {2nd},
   url = {https://store.doverpublications.com/0486151034.html},
   year = {2003}
}

@book{Olver95,
   author = {Olver, P.J.},
   title = {Equivalence, Invariants, and Symmetry},
   publisher = {Cambridge Univ. Press},
   url = {https://doi.org/10.1017/CBO9780511609565},
   year = {1995}
}

@ARTICLE{Boozer_pf_1981,
   author       = "A.H. Boozer",
   title        = "Plasma equilibrium with rational magnetic surfaces",
   year         = "1981",
   journal      = "Phys. Fluids",
   volume       = "24",
   pages        = "1999--2003",
   number    ="11",
   url = "https://doi.org/10.1063/1.863297"
}

@ARTICLE{Hamada_1962,
   author       = "S. Hamada",
   title        = "Hydromagnetic equilibria and their proper coordinates",
   year         = "1962",
   journal      = "Nucl. Fusion",
   volume       = "2",
   pages        = "23",
   url =  "https://doi.org/10.1088/0029-5515/2/1-2/005",
   number    =""
}

@article{Perrella23,
   author = {Perrella, D. and Duignan, N. and Pfefferlé, D.},
   title = {Existence of Global Symmetries of Divergence-Free Fields with First Integrals},
   journal = {J. Math. Phys.},
   volume = {64},
   number = {5},
   pages = {052705},
   url = {https://doi.org/10.1063/5.0152213},
   year = {2023}
}

@article{Helander14,
   author = {Helander, P.},
   title = {Theory of Plasma Confinement in Non-Axisymmetric Magnetic Fields},
   journal = {Reports on Progress in Physics},
   volume = {77},
   number = {8},
   pages = {087001},
   url = {https://doi.org/10.1088/0034-4885/77/8/087001},
   year = {2014}
}

@book{Imbert25,
   author = {Imbert-Gérard, L-M and Paul, E. and Wright, A.},
   title = {An Introduction to Stellarators},
   publisher = {SIAM},
   series = {Other Titles in Applied Mathematics},
   url = {https://doi.org/10.1137/1.9781611978223},
   year = {2025}
}

@article{boozerTransportIsomorphicEquilibria1983,
  title = {Transport and Isomorphic Equilibria},
  author = {Boozer, A.H.},
  year = 1983,
  month = feb,
  journal = {Phys. Fluids},
  volume = {26},
  number = {2},
  pages = {496--499},
  issn = {0031-9171},
  url = {https://doi.org/10.1063/1.864166},
  urldate = {2026-08-25},
}

@inproceedings{goriQuasiisodynamicStellarators1997,
  title = {Quasi-Isodynamic Stellarators},
  booktitle = {Theory of {{Fusion Plasmas-Proceedings}} of the {{Joint Varenna-Laussane International Workshop}}},
  author = {Gori, S. and Lotz, W. and N{\"u}hrenberg, J.},
  year = 1997,
  pages = {335--342},
  urldate = {2023-11-02}
}

@article{caryOmnigenityQuasihelicityHelical1997,
  title = {Omnigenity and Quasihelicity in Helical Plasma Confinement Systems},
  author = {Cary, J.R. and Shasharina, S.G.},
  year = 1997,
  month = sep,
  journal = {Physics of Plasmas},
  volume = {4},
  number = {9},
  pages = {3323--3333},
  issn = {1070-664X},
  url = {https://doi.org/10.1063/1.872473},
  urldate = {2026-08-25},
  }

@article{landremanDirectConstructionOptimized2018,
  title = {Direct Construction of Optimized Stellarator Shapes. {{Part}} 1. {{Theory}} in Cylindrical Coordinates},
  author = {Landreman, M. and Sengupta, W.},
  year = 2018,
  month = dec,
  journal = {J. Plas. Phys.},
  volume = {84},
  number = {6},
  publisher = {Cambridge University Press},
  issn = {0022-3778, 1469-7807},
  url = {https://doi.org/10.1017/S0022377818001289},
  }

@article{landremanMagneticFieldsPrecise2022,
  title = {Magnetic {{Fields}} with {{Precise Quasisymmetry}} for {{Plasma Confinement}}},
  author = {Landreman, M. and Paul, E.},
  year = 2022,
  month = jan,
  journal = {Phys. Rev. Lett.},
  volume = {128},
  number = {3},
  pages = {035001},
  publisher = {Amer Physical Soc},
  address = {College Pk},
  issn = {0031-9007, 1079-7114},
  url = {https://doi.org/10.1103/PhysRevLett.128.035001},
  }

@article{goodmanConstructingPreciselyQuasiisodynamic2023,
  title = {Constructing Precisely Quasi-Isodynamic Magnetic Fields},
  author = {Goodman, A.G. and Mata, K. Camacho and Henneberg, S. A. and Jorge, R. and Landreman, M. and Plunk, G. G. and Smith, H. M. and Mackenbach, R. J. J. and Beidler, C. D. and Helander, P.},
  year = 2023,
  month = sep,
  journal = {J. Plas. Phys.},
  volume = {89},
  number = {5},
  pages = {905890504},
  publisher = {Cambridge Univ Press},
  address = {Cambridge},
  issn = {0022-3778, 1469-7807},
  url = {https://doi.org/10.1017/S002237782300065X},
  }

@article{dudtMagneticFieldsGeneral2024,
  title = {Magnetic Fields with General Omnigenity},
  author = {Dudt, D.W. and Goodman, A.G. and Conlin, R. and Panici, D. and Kolemen, E.},
  year = 2024,
  month = feb,
  journal = {J. Plas. Phys.},
  volume = {90},
  number = {1},
  pages = {905900120},
  issn = {0022-3778, 1469-7807},
  url = {https://doi.org/10.1017/S0022377824000151},
  }

@article{clellandBeltramiFieldsNonconstant2020,
  title = {Beltrami {{Fields}} with {{Nonconstant Proportionality Factor}}},
  author = {Clelland, J.N. and Klotz, T.},
  year = 2020,
  month = may,
  journal = {Archive for Rational Mechanics and Analysis},
  volume = {236},
  number = {2},
  pages = {767--800},
  issn = {0003-9527, 1432-0673},
  url = {https://doi.org/10.1007/s00205-019-01481-7},
}

@article{burbyCharacterizationAdmissibleQuasisymmetries2025,
  title = {Characterization of Admissible Quasisymmetries},
  author = {Burby, J.W. and Kallinikos, N. and MacKay, R.S. and Perrella, D. and Pfefferl{\'e}, D.},
  year = 2025,
  month = feb,
  journal = {J. Plas. Phys.},
  volume = {91},
  number = {1},
  pages = {E33},
  issn = {0022-3778, 1469-7807},
  url = {https://doi.org/10.1017/S0022377824001478},
  urldate = {2025-03-03},
}

@article{senguptaStellaratorEquilibriumAxisexpansion2024,
    author = {Sengupta, W. and Rodriguez, E. and Jorge, R. and Landreman, M. and Bhattacharjee, A.},
    title = {Stellarator equilibrium axis-expansion to all orders in distance from the axis for arbitrary plasma beta},
    volume = {90},
    issn = {0022-3778, 1469-7807},
    url = "https://doi.org/10.1017/S002237782400076X",
    number = {4},
    urldate = {2026-09-01},
    journal = {J. Plas. Phys.},
    month = aug,
    year = {2024},
    pages = {905900407},
}

@article{jorgeNearaxisExpansionStellarator2020,
    title = {Near-axis expansion of stellarator equilibrium at arbitrary order in the distance to the axis},
    volume = {86},
    copyright = {https://www.cambridge.org/core/terms},
    issn = {0022-3778, 1469-7807},
    url = {https://www.cambridge.org/core/product/identifier/S0022377820000033/type/journal_article},
    doi = {10.1017/S0022377820000033},
    number = {1},
    urldate = {2026-09-01},
    journal = {J. Plas. Phys.},
    author = {Jorge, R. and Sengupta, W. and Landreman, M.},
    month = feb,
    year = {2020},
    pages = {905860106},
}

@article{landremanMappingSpaceQuasisymmetric2022,
    title = {Mapping the space of quasisymmetric stellarators using optimized near-axis expansion},
    volume = {88},
    issn = {0022-3778, 1469-7807},
    url = {https://www.cambridge.org/core/product/identifier/S0022377822001258/type/journal_article},
    doi = {10.1017/S0022377822001258},
    number = {6},
    urldate = {2026-09-01},
    journal = {J. Plas. Phys.},
    author = {Landreman, M. },
    month = dec,
    year = {2022},
    pages = {905880616},
}

@article{garrenExistenceQuasihelicallySymmetric1991,
    title = {Existence of quasihelically symmetric stellarators},
    volume = {3},
    issn = {0899-8221},
    url = {https://pubs.aip.org/pfb/article/3/10/2822/841899/Existence-of-quasihelically-symmetric-stellarators},
    doi = {10.1063/1.859916},
    number = {10},
    urldate = {2025-07-02},
    journal = {Phys. Fluids B: Plasma Physics},
    author = {Garren, D.A. and Boozer, A.H.},
    month = oct,
    year = {1991},
    pages = {2822--2834},
}

\newpage
\appendix
\begin{center}
\Large{\textbf{Appendices}}
\end{center}

\section{Solution to an Example PDE} \label{app:Example}
To solve the example system \Eq{CK-Example}, we will establish a sequence of two Cauchy problems. The initial data can be specified by two arbitrary functions of one variable and two arbitrary functions of two variables.
Explicitly, the data is given first for $u^1$ and $u^2$ along the curve  $x^2 = x^3 =0$ as
\[
    u^1(x^1,0,0)=f^1(x^1) \;, \text{ and }u^2(x^1,0,0)=f^2(x^1) \;,
\]
and then for $u^3$ and $u^4$, on the surface $x^3 = 0$, as
\[
    u^3(x^1,x^2,0)=g^3(x^1,x^2)\;, \text{ and } u^4(x^1,x^2,0)=g^4(x^1,x^2),
\]
for arbitrary analytic functions $f^1$, $f^2$, $g^3$ and $g^4$. 

Then, to solve the example system \Eq{CK-Example}, we first seek to find the evolution along the $x^2$ direction for $u^1$ and $u^2$, setting $x^3 = 0$. Let 
\[
    u^1(x^1,x^2,0)=g^1(x^1,x^2)\text{ and }u^2(x^1,x^2,0)=g^2(x^1,x^2)
\]
be the solution to this Cauchy problem which can be written
\begin{equation}\label{eq:CK-x2-ev}
\begin{split}
    \partiald{g^1}{x^2}&=g^2-g^1+g^3g^4 \;,\\
    \partiald{g^2}{x^2}&=g^1-g^2 \;,
\end{split}
\quad \text{ with }\quad
\begin{split}
    g^1(x^1,0)&=f^1(x^1) \;,\\
    g^2(x^1,0)&=f^2(x^1) \;,
\end{split}
\end{equation}
i.e., the first two equations of \Eq{CK-Example} for $x^3=0$. Since $g^3$ and $g^4$ are given, this is a system of linear ODEs for $g^1$ and $g^2$, and the solution is easily found. We can take a simple linear combination of the the PDEs \Eq{CK-x2-ev} to obtain
\begin{align*}
   \frac{\partial}{\partial x^2}\left( g^1+g^2\right) &=g^3g^4 \;, \\
   \frac{\partial}{\partial x^2}\left( g^2-g^1\right) &=-2(g^2-g^1)-g^3g^4 \;.
\end{align*}
Using the initial conditions \Eq{CK-x2-ev}, the solution to these ODEs is:
\begin{align*}
    g^1 +g^2 &= f^1(x^1)+f^2(x^1)+I_1(x^1,x^2) \;, \\
    g^2 -g^1 &=e^{-2x^2} \left[f^2(x^1)-f^1(x^1) -I_2(x^1,x^2) \right] \;,
\end{align*}
where 
\begin{align*}
    I_1(x^1,x^2) &=\int_0^{x^2}g^3(x^1,\eta)g^4(x^1,\eta)\,d\eta \;, \\
    I_2(x^1,x^2) &=\int_0^{x^2}e^{2\eta}g^3(x^1,\eta)g^4(x^1,\eta)\,d\eta \;.
\end{align*}
Therefore the solution to \Eq{CK-x2-ev} is
\begin{align*}
    g^1 &=\tfrac{1}{2}\left(f^1+f^2-(f^2-f^1)e^{-2x^2}+I_1+e^{-2x^2}I_2\right) \;, \\
    g^2 &=\tfrac{1}{2}\left(f^1+f^2+(f^2-f^1)e^{-2x^2}+I_1-e^{-2x^2}I_2\right) \;.
\end{align*}

This solution, in conjunction with $g^3$ and $g^4$, specifies the Cauchy data for the evolution in $x^3$, from \Eq{CK-Example}, which becomes
\begin{equation}\label{eq:CK-x3-ev}
\begin{split}
    \partiald{u^1}{x^3}&=1 \;,\\
    \partiald{u^2}{x^3}&=1 \;,\\
    \partiald{u^3}{x^3}&=u^3 \;,\\
    \partiald{u^4}{x^3}&=-u^4 \;,
\end{split} 
\quad \text{ with }\quad
\begin{split}
    u^1(x^1,x^2,0)&=g^1(x^1,x^2) \;,\\
    u^2(x^1,x^2,0)&=g^2(x^1,x^2) \;,\\
    u^3(x^1,x^2,0)&=g^3(x^1,x^2) \;,\\
    u^4(x^1,x^2,0)&=g^4(x^1,x^2) \;.
\end{split}
\end{equation}
Consequently the initial data for \Eq{CK-x3-ev} is determined by the initial values $g^3$ and $g^4$ together with the solution to the Cauchy problem \Eq{CK-x2-ev}. This Cauchy problem is straightforward to solve with the (now known) initial data $g^1$ and $g^2$:
\begin{equation}\label{eq:CK-ex-sol}
\begin{split}
    u^1(\bm{x})&=x^3+g^1(x^1,x^2),\\
    u^2(\bm{x})&=x^3+g^2(x^1,x^2),
\end{split}
\qquad
\begin{split}
    u^3(\bm{x})&=g^3(x^1,x^2)e^{x^3},\\
    u^4(\bm{x})&=g^4(x^1,x^2)e^{-x^3}.
\end{split}
\end{equation}

Note that to guarantee that \Eq{CK-ex-sol} are full solutions to \Eq{CK-Example} we need to check for compatibility, i.e., that these solutions also give all the possible solutions to \Eq{CK-x2-ev}. While  this compatibility condition is not generally obvious, in this case, it is simple. First notice that, critically, $u^3u^4=g^3g^4$. Now since 
$\partial_{x^2}u^1=\partial_{x^2}(x^3+g^1)= \partial_{x^2}g^1$ and 
$\partial_{x_2}u^2=\partial_{x^2}(x^3+g^2)= \partial_{x^2}g^2$,
we reduce back to \Eq{CK-x2-ev}, and we are done.

\section{Cartan-K\"ahler using Differential Forms}\label{app:CK-forms-ver}
Consider a PDE system of $r$ equations with $\ell$ unknowns $\bm{u} = (u^1,\ldots,u^\ell)$ depending upon $n$ variables $\bm{x} = (x^1,\ldots,x^n)$, given by
\begin{equation}\label{eq:PDE-system}
    \cF^j(\bm{x},\bm{u}, D\bm{u}) = 0, \quad j = 1, \ldots, r
\end{equation}
Assume that on some open set in $\bR^{n\times\ell}$ we can solve this system to obtain a set of $r$ first order PDEs
\[
    \partiald{u^a}{x^j} = G^a_j(\bm{x},\bm{u},\widetilde{D\bm{u}}) \;,
\]
for $r$ entries of the Jacobian matrix $D\bm{u}$ and where $\widetilde{D\bm{u}}$ are the $n\ell-r$ remaining entries of the Jacobian.

The Cartan-K\"ahler algorithm determines existence of local analytic solutions to \Eq{PDE-system}, using the following five steps (the first four of where given in \Sec{Cartan-Kahler}):
\begin{enumerate}
\item[(1)] \textbf{Convert a PDE to an Exterior Differential System denoted $\cI$:}
For first order PDE systems we need only work on a manifold $(\bm{x},\bm{u},\bm{p}) \in  \bR^n \times \bR^\ell \times \bR^{n\ell} \equiv \widehat{M}$ (in particular a submanifold which we specify shortly).
Here the added coordinates represent the first derivatives $p^a_i=\partial_{x^i} u^a$ of the dependent variables. 
The EDS is generated by  the differential 1-forms
\[
    \alpha^a=du^a-p^a_1\,dx^1-\cdots-p^a_n\,dx^n, \quad a = 1,\ldots, \ell \;.
\]
In these equations we replace the the $p^a_i$'s that are part of the PDE system \Eq{PDE-system} by the expressions from this system. The remaining $p^a_i$ are place holders for derivatives not occurring in the PDE.
 The implication is that we are now working on a smaller manifold $M=\bR^n\times \bR^\ell\times \bR^{n\ell-r}\subset \widehat{M}$ where $r$ is the number of $p^a_i$ defined by the PDE system. That is, $M$ is defined exactly by the equations $\cF^j(\bm{x},\bm{u}, \bm{p})=0$. 
Solutions to the PDE system live on a submanifold $N$ on which 
\begin{equation}\label{eq:EDS-Forms-version}
    \alpha^a=0 \quad \mbox{and}  \quad  d\alpha^a=0 ,\quad \quad  a = 1, \ldots, \ell. 
\end{equation}
These equations can be thought of as a \textit{differential ideal}
\[
    \cI = \langle \alpha^1, \alpha^2, \ldots, \alpha^\ell \rangle .
\]
It is also helpful to speak only of the 1-form generators, which we denote as $I=\{\alpha^1,\ldots,\alpha^\ell\}$ using the notation $\cI=\langle I \rangle$. Here $\cI$ is an ideal in the sense that all linear combinations and wedge products of the forms in $\cI$ as well as all of their differentials must vanish on graphs of solutions.
The collection $\cI$, together with all of its exterior derivatives and wedge products, is called an \textit{exterior differential system} (EDS). 

Note that in some cases PDEs may be written using $k$-forms with $k>1$ (these might represent, e.g., curl operators), this does not harm the generality of our presentation since, at the cost of adding more derivative variables, we can always arrive at a system described by $1$-forms. This process is called \textit{prolongation}. For example, this is used in \Lem{prolonged-EDS} to obtain the EDS describing the realizability problem in \Sec{ProlongEDS}. 

\textbf{Example:}
Converting equations \Eq{CK-Example}, with $n=3$ and $\ell=4$, into an EDS gives \Eq{EDS-CK-Example}. Note that $p^1_{2}$, $p^1_{3}$, $p^2_{2}$, $p^2_{3}$, $p^3_3$ and $p^4_3$ have been replaced by the six expressions for $\partial_{x^i}u^a$ from the PDE system \Eq{CK-Example}. The matrix $\Gamma$ of \Sec{CK-Example} has entries given precisely by the coefficients of each $dx^i$ in \Eq{EDS-CK-Example}, with the zero entries arising from the extension to the full space $\mathbb{R}^3\times\mathbb{R}^{10}$. 

\item[(2)]\textbf{Differentiate to find Tableau and Torsion:}
Given the EDS $\cI$,  it is convenient to generalize the one forms $\{dx^1,\ldots,dx^n\}$ to a new linearly independent set that we call $\{\omega^1,\ldots,\omega^n\}$.

The point is that the general algorithm does not require that the $\omega^i$ represent the exterior derivatives of the independent variables exactly, and the freedom in the choice can be useful in finding a `natural' coordinate system for the PDE. This notation is used in the proof of \Th{FTSD} in \Sec{ProlongEDS}. 

Differentiation of the EDS \Eq{EDS-Forms-version} gives a system of 2-forms
\begin{equation}\label{eq:pre-struct-alg}
    d\alpha^a= \Pi^a_i\w  \omega^i
        + T^a_{jk}\,\omega^j\w  \omega^k +\gamma^a_b\w \alpha^b , \quad a = 1, \ldots, \ell \;\,
\end{equation}
where we use the summation convention to sum over $i$, $j$, $k$, and $b$.
Here $\Pi$ is an $\ell \times n$ matrix of 1-forms, called the \textit{tableau}, and the $\ell$ different $n \times n$ matrices with elements $T^a_{jk}$ are called the \textit{torsion}.
As in \Eq{AlphaStructure}, the entries of the tableau can be written as $\Pi^a_i=A^a_{\epsilon i}\pi^\epsilon$, where $\pi^\epsilon$ are $1$-forms that arise from the differentiation that are linearly independent from the $\omega^i$ and the $\alpha^a$, but depend upon the $p^\epsilon_i$, e.g., something like $\pi^1=dp^1_1+dp^1_2-3\omega^2$. In particular, the $\pi^\epsilon$ may have components in the $\omega^i$ directions when expressed in terms of $dp^a_i$'s.

We mention here that the $\Pi$ in this appendix captures all the same information as the $\Pi$ in \Sec{Linear_PDEs} and \Sec{Pfaffian_Tableau}, but the two notations are slightly different, which we discuss below in reference to the example.  

Since solutions \Eq{EDS-Forms-version} correspond to $\alpha^a = 0$, the terms $\gamma^a_b\w  \alpha^b$ in \Eq{pre-struct-alg} will actually contribute nothing to the analysis and can be safely ignored. In other words  we can rewrite \Eq{pre-struct-alg} as
\begin{equation}\label{eq:pre-struct-mod}
    d\alpha^a\equiv \Pi^a_i\w  \omega^i 
        + T^a_{jk}\,\omega^j\w  \omega^k \mod \alpha^1, \ldots, \alpha^\ell \;,
\end{equation}
using equivalence, $\equiv$, to denote setting the $\alpha^i$ to zero. We will often use ``mod $I$" or ``mod $\alpha$" to represent this operation. \\

\textbf{Example:} 
For the example of \Sec{CK-Example}, the exterior derivatives of \Eq{EDS-CK-Example} are
\begin{equation}\label{eq:EDS-CK-Example-pre-struct}
\begin{aligned}
    d\alpha^1&=-dp^1_1\w  dx^1-(du^2-du^1+u^4\,du^3+u^3\,du^4)\w  dx^2 \;,\\
    d\alpha^2&=-dp^2_1\w  dx^1-(du^1-du^2)\w  dx^2 \;,\\
    d\alpha^3&=-dp^3_1\w  dx^1-dp^3_2\w  dx^2-du^3\w  dx^3 \;,\\
    d\alpha^4&=-dp^4_1\w  dx^1-dp^4_2\w  dx^2+du^4\w  dx^3 \;.
\end{aligned}
\end{equation}
Using \Eq{EDS-CK-Example} to replace each $du^a$ with terms involving $\alpha^a$ and $dx^i$,
and then dropping the $\alpha^a$ terms with ``mod'' gives  
\begin{equation}\label{eq:EDS-CK-Example-struct}
\begin{split}
    d\alpha^1&\equiv -dp^1_1\w  dx^1-(p^3_1u^4+p^4_1u^3+p^2_1-p^1_1)\,dx^1\w  dx^2  \\
    d\alpha^2&\equiv-dp^2_1\w  dx^1-(p^1_1-p^2_1)\,dx^1\w  dx^2\\
    d\alpha^3&\equiv-dp^3_1\w  dx^1-dp^3_2\w  dx^2-p^3_1\,dx^1\w \,dx^3-p^3_2\,dx^2\w  dx^3\\
    d\alpha^4&\equiv-dp^4_1\w  dx^1-dp^4_2\w  dx^2+p^4_1\,dx^1\w  dx^3+p^4_2\,dx^2\w  dx^3
\end{split}
\begin{split}
    \quad \mod \alpha \;.
\end{split}
\end{equation}
Writing these equations in the matrix form \Eq{pre-struct-mod} then gives
\begin{equation}\label{eq:EDS-CK-Example-struct-matrix}
\begin{bmatrix}
    d\alpha^1\\ 
    d\alpha^2\\ 
    d\alpha^3\\
    d\alpha^4
\end{bmatrix} 
    \equiv \begin{bmatrix} 
        -dp^1_1 & 0 & 0 \\ 
        -dp^2_1 & 0 & 0 \\ 
        -dp^3_1 & -dp^3_2 & 0 \\ 
        -dp^4_1 & -dp^4_2 & 0 
    \end{bmatrix}\w  
    \begin{bmatrix} 
        \omega^1 \\ \omega^2 \\ \omega^3 
    \end{bmatrix}
    +\begin{bmatrix}
        -(p^3_1u^4+p^4_1u^3+p^2_1-p^1_1)\,\omega^1\w  \omega^2\\
        -(p^1_1-p^2_1)\,\omega^1\w  \omega^2\\
        -p^3_1\,\omega^1\w \,\omega^3-p^3_2\,\omega^2\w  \omega^3\\
        p^4_1\,\omega^1\w  \omega^3+p^4_2\,\omega^2\w  \omega^3 
    \end{bmatrix} \mod \alpha \;,
\end{equation}
where $\omega^i=dx^i$ and the `$\w $' in the first matrix multiplication means that the products on entries is replaced by the wedge product. The tableau $\Pi^a_i$ corresponds to the $4\times 3$ matrix of $1$-forms  whose entries are the $dp^a_i$. The torsion $T^a_{ij}$ corresponds to  the vector of $2$-forms in the $\omega^i$. Notice that in this example, the tableau matrices $A^a$ are very simple---each non-zero entry corresponding to a $\pi^\epsilon$ contains a single $dp^a_i$. 

In relation to \Sec{CK-Example}, notice that the $C^1,\ldots,C^4$ matrices are skew symmetric, and that their entries are captured by each of the entries of the $T^1_{ij}\omega^i\wedge\omega^j$ vectors  in \Eq{EDS-CK-Example-struct-matrix}, so that $C^1$ corresponds exactly to the functions $T^1_{ij}$, and similarly for $C^2$, etc. 

\item[(3)] \textbf{Eliminate Torsion:}
Since solutions to the PDE system \Eq{EDS-Forms-version} must have $d\alpha^a = 0$, the torsion functions $T^a_{ij}$ form an obstruction to existence of solutions: they capture conditions under which partial derivatives may fail to commute between the $\omega^i$ and $\omega^j$ directions, for $i\neq j$. 

As such, we must eliminate the torsion as much as possible. This can be done by \textit{absorbing} these terms into the tableau. To do this we can modify the linearly independent (in particular, non-zero) $\pi^\epsilon$ to give new tableau $1$-forms, i.e. $\pi^\epsilon \to \tilde{\pi}^\epsilon$.  
There are two possible outcomes:
\begin{enumerate}
        \item All of the $T^a_{jk}$ terms can be ``absorbed" into the tableau. In this case we call the torsion \textit{absorbable}. 
        \item After absorbing as much torsion as possible, there still remain some $T^a_{jk}$ terms. In this case the remaining terms impose new algebraic constraints or differential equations. One must then start over by modifying the EDS $\cI$ in step (2). This process continues as long as some of the torsion cannot be absorbed. 
\end{enumerate} 
If the torsion cannot be ultimately absorbed, then either the system has only the zero solution or else there are no solutions that satisfy all the additional equations. 

The goal in either case is to reduce \Eq{pre-struct-mod} to the form
\begin{equation}\label{eq:EDS-reduced}
         d\alpha \equiv \tilde{\Pi}\w  \omega \mod \alpha
\end{equation}
With the absorbed torsion one can move onto the next step with the modified tableau. 

There is an abstract linear algebra formulation of this process, which is partially presented \Sec{Cartan-Kahler}.
Let $J=\{\alpha^1,\ldots,\alpha^\ell,\omega^1,\ldots,\omega^n\}$ and $J/I=\{\omega^1,\ldots,\omega^n\}$. The tableau $\Pi$ can be understood pointwise as a map $\Pi:J^\perp \to I^*\otimes(J/I)$. Denote the vector subspace $A=\text{Im}(\Pi)$ (this $A$ is also referred to as the tableau) and define the skew-symmetrization map
$\sigma: I^* \otimes (J/I)\otimes(J/I)\to I^*\otimes((J/I\wedge(J/I))$ where $(J/I)\wedge(J/I)$ is the degree 2 exterior algebra of $J/I$, and is what captures the skew-symmetrization. Define 
\begin{equation}\label{eq:Spencer-Cohomology}
    H^{0,2}(A)=\frac{I^*\otimes((J/I)\wedge(J/I))}{\sigma(A\otimes (J/I))}.
\end{equation}
Since the torsion terms $T^a_{ij}\omega^i\wedge \omega^j$ can be interpreted as elements of $I^*\otimes ((J/I)\wedge(J/I))$ then they admit a representative $[T]\in H^{0,2}(A)$. In particular, torsion is absorbable if $H^{0,2}(A)=0$. The algebraic object $H^{0,2}(A)$ is a \textit{Spencer cohomology group}, which can encode another approach to involutivity \cite{EDS-Book}.

\textbf{Example:}
For the PDE system \Eq{CK-Example}, the equations \Eq{EDS-CK-Example-struct-matrix} appear to have non-trivial torsion terms. However, we can actually absorb all of this torsion into the tableau. Indeed, notice that we can re-write \Eq{EDS-CK-Example-struct-matrix} as 
\begin{equation}\label{eq:EDS-CK-Example-struct-matrix-no-tor}
    \begin{bmatrix}
        d\alpha^1\\ 
        d\alpha^2\\ 
        d\alpha^3\\
        d\alpha^4
    \end{bmatrix} = \begin{bmatrix} -dp^1_1+(p^3_1u^4+p^4_1u^3+p^2_1-p^1_1)\omega^2 & 0 & 0 \\ -dp^2_1+(p^1_1-p^2_1)\,\omega^2 & 0 & 0 \\ -dp^3_1+p^3_1\,\omega^3 & -dp^3_2+p^3_2\,\omega^3 & 0 \\ -dp^4_1-p^4_1\,\omega^3 & -dp^4_2-p^4_2\,\omega^3 & 0  \end{bmatrix}\w  \begin{bmatrix} \omega^1 \\ \omega^2 \\ \omega^3 \end{bmatrix} \;,
\end{equation}
which gives a new tableau $\tilde{\Pi}$ given by the above $4\times 3$ matrix of 1-forms. Comparing against the presentation in \Sec{CK-Example}, we see that the particular solution given to $\hat{\Sigma}p^\circ=C$ is captured by the coefficients of the $\omega^i$ modifications of the $dp^a_i$ entries of $\Pi$ that yield $\tilde{\Pi}$ e.g. the $v^5_2$ of equation 
\Eq{Sigma-hat_image_ex} is exactly the coefficient of $\omega^2$ in $\tilde{\Pi}^1_1$ of \Eq{EDS-CK-Example-struct-matrix-no-tor}. Generally, the explicit skew-symmetrization captured by a particular solution of $\hat{\Sigma}p^\circ=C$ precisely encodes the wedge product compatibility of the torsion terms $T^a_{ij}\omega^i\wedge \omega^j$ with the ``shape" of the matrix of 1-forms $\Pi$. 

For simplicity of notation, we will abuse notation by dropping the tilde so that $\tilde{\Pi} \to \Pi$.

\item[(4)] \textbf{Apply Cartan's Test to the Tableau:} 
The tableau is fundamentally a linear algebraic object that captures a pointwise linearization of the PDE. Cartan developed a test that can be used on $\Pi$ to determine if the PDE admits local analytic solutions via a sequence of Cauchy problems. Here is a \textit{quick version} of this test.
\begin{enumerate}
    \item For the tableau $\Pi$ in \Eq{EDS-reduced}, let $s_k$ be the number of linearly independent $1$-forms in column $k$ that are also linearly independent of the $1$-forms in the previous $k-1$ columns of $\Pi$. Each $s_k$ is called a \textit{Cartan Character}. \\
    
\textbf{Example:} For the example \Eq{EDS-CK-Example-struct-matrix-no-tor}, we see that the first column of $\Pi$ has all entries linearly independent and so $s_1=4$. Similarly the two entries in the second column of $\Pi$ are linearly independent from each other as well as the entries in the $1^{st}$ column and so $s_2=2$. Finally, we see that $s_3=0$ since the 3rd column of $\Pi$ has only zero entries. 

    \item Compute the degrees of freedom in the tableau, which we denote as $s^P.$ 
    To do this, we need to find the dimensions of the kernel of a skew-symmetrization induced by the wedge product that only incorporates the 1-forms $\omega^i$. That is, we
    create a matrix of $1$-forms $\Lambda=[\lambda^a_i]$ in the same ``shape" as $\Pi$ (i.e. in the same subspace of $s\times n$ matrices generated by $\Pi$) by modifying the  $\Pi^a_i$ that are linearly independent by  
    \[
        \Pi^a_i \mapsto \Pi^a_i+\lambda^a_i=\Pi^a_i+h^a_{ij}\omega^j. 
    \]
    The number, $s^P,$ of free variables $h^a_{ij}$ from solving the matrix equations
     \[
        (\Pi+\Lambda)\w  \omega \equiv \Pi\w  \omega \mod \alpha
    \]
    is called the \textit{dimension of the prolongation} or sometimes the degree of indeterminacy. Note that $s^P$ is exactly $\dim\left(\Pi^{(1)}\right)$ from \Sec{check_tableau_example}. \\
    
    \textbf{Example:}
    Let $\Lambda$  be the $4\times 3$ matrix with non-zero entries denoted by $\lambda^a_i$. 
    Then 
    \[
        \Pi+\Lambda =\begin{bmatrix} \lambda^1_1-dp^1_1+(p^3_1u^4+p^4_1u^3+p^2_1-p^1_1)\omega^2 & 0 & 0 \\ \lambda^2_1-dp^2_1+(p^1_1-p^2_1)\,\omega^2 & 0 & 0 \\ \lambda^3_1-dp^3_1+p^3_1\,\omega^3 & \lambda^3_2-dp^3_2+p^3_2\,\omega^3 & 0 \\ \lambda^4_1-dp^4_1-p^4_1\,\omega^3 & \lambda^4_2-dp^4_2-p^4_2\,\omega^3 & 0  \end{bmatrix}
    \]
    So to have $(\Pi+\Lambda)\w \omega = \Pi\w  \omega$ in this case, it must be that 
    \[
        \Lambda = \begin{bmatrix} h^1_{11}\omega^1 & 0 & 0\\ h^2_{11}\omega^1 & 0 & 0\\ h^3_{11}\omega^1+h^3_{12}\omega^2 & h^3_{21}\omega^1+h^3_{22}\omega^2 & 0 \\
        h^4_{11}\omega^1+h^4_{12}\omega^2 & h^4_{21}\omega^1+h^4_{22}\omega^2 & 0 
    \end{bmatrix}
    \]
    where $h^3_{12}=h^3_{21}$ and $h^4_{12}=h^4_{21}$ so that when we wedge with $\omega$ these terms cancel and leave $\Pi$ unchanged. As such, we count 8 free variables $h^a_{ij}$ and hence $s^P=8$.\\
    
    We can compare this situation to \Sec{check_tableau_example}. Indeed, the dimension of $\dim\ker \hat{\Sigma}$ is precisely $s^P$ and the calculation of $\Lambda$ is equivalent to the calculation of $\ker \hat{\Sigma}$. In fact, from an abuse of notation, the $\Pi$ of \Sec{check_tableau_example} is precisely the image of the matrix of 1-forms $\Pi$ presented in this appendix when thought of as the linear map $\Pi:\mathbb{R}^6 \to \mathbb{R}^{4\times 3}$ defined by how the 1-form entries act as linear functionals on $\mathbb{R}^6$. These 1-form entries are precisely the $\pi^a_i$'s mentioned in \Sec{Pfaffian_Tableau} and defined in the language of that section by \Eq{tableau_entries_general}. \\
    
    We also remark that in \Sec{Proof} both the torsion absorption step and the dimension of prolongation calculations are carried out by analyzing a matrix $R$ as a linear map which arises from exactly $\Lambda$ and is closely related to a representation of the map $\sigma$ for the Spencer cohomology of equation \Eq{Spencer-Cohomology} as well as the map $\hat{\Sigma}$ of \Sec{Pfaffian_Torsion}. 
    \item Finally, we check the inequality
    \begin{equation}\label{eq:Cartan-Test}
        s^P\leq s_1+2\,s_2+\cdots+n\,s_{n} .
    \end{equation}
    If \Eq{Cartan-Test} is actually an \underline{equality}, then we say the tableau is \textit{involutive} and we are guaranteed that the original PDE system has local analytic solutions and that initial data for each Cauchy problem depends on successive differences of the non-zero Cartan characters. 
    
    \textbf{Example:} We have that $s_1=4$, $s_2=2$, $s_3=0$, and $s^P=8$ and so equality of \Eq{Cartan-Test} is verified. Notice that to solve the last Cauchy problem in the sequence, one needs 4 functions of 2 variables, but 2 of those 4 are already determined by the solutions to the previous Cauchy problem and hence solutions to the whole PDE depends on initial data determined by $s_2=2$ functions of 2 variables, and $s_2-s_1=2$ functions of 1 variable.
\end{enumerate}
\item[(5)] \textbf{Add Equations if Step (4) Fails:} If the inequality \Eq{Cartan-Test} is strict, then the system is not involutive, and one must now add additional $1$-forms to $\cI$ on a larger manifold of dimension at least $s^P$ larger and start the entire process over again. This is called prolongation.\footnote
{On rare occasions, one gets a false negative for failure of equality for Cartan's test. This can be a subtle matter but is usually handled by modifying the $\omega^i$'s. A more reliable method, replacing step (4a), can be found in  \cite{Olver95}.}
The Cartan-Kuranishi theorem \cite{EDS-Book} guarantees that by adding a finite number of additional differential equations (a process called prolongation) we can either express the original PDE plus the additional derivative conditions into Cauchy-Kowalevsky form, or the PDE admits no non-trivial analytic solutions. This is a landmark result in exterior differential systems theory. 
\end{enumerate}

\section{An adapted frame for a sphere bundle}\label{app:Convenient_Frame}
Suppose that $Q$ is an oriented Riemannian 3-manifold with metric $g$, and let $\Omega$ denote its Riemannian volume form. The unit sphere bundle over $Q$ is the five dimensional fiber bundle with fibers
\[
    S_{\bm{x}} Q \equiv \{\bm{b}\in T_{\bm{x}}Q \, :\, |\bm{b}|^2 = 1 \} \;.
\]
We denote the projection onto the base space by $\pi_S: SQ \to Q$.
We will build a convenient coframe on a suitable open subset $U\subset SQ$ for which the forms $\beta^1,\beta^2$ defined in \Eq{BetaForms} have relatively simple descriptions. In particular, the coframe is given by $\{\bar\omega^1, \bar\omega^2, \bar\omega^3, \bar\omega^3_1, \bar\omega^3_2\}$ whose elements have the properties
\begin{equation}\label{eq:omegaProp}
    \bar\omega^3 = \bm{b}^\flat \mbox{ and } 
    \bar\omega^1 \w \bar\omega^2 = \iota_{\bm{b}} \Omega \;,
\end{equation}
so that technically, as forms on $SQ$, we have that
\[ 
    \bm{b}^\flat (\cdot) \equiv g|_{\bm{x}} (\bm{b}, d\pi_S(\cdot)), \quad 
    (\iota_{\bm{b}} \Omega)|_{(\bm{x},\bm{b})} (\cdot,\cdot) 
    \equiv  \pi_S^* \iota_{\bm{b}}\Omega \;.
\]
The remaining coframe elements will be chosen so that the exterior derivatives (that is, the structure equations) are also relatively simple, as shown in \Lem{fExists} below.

To construct the coframe, we will consider the bundle of orthonormal frame $FQ$ as a bundle over $SQ$. A local section of this bundle identifies an open subset of $SQ$ with a submanifold of $FQ$, allowing certain canonical forms on $FQ$ to be pulled back to $SQ$. 

Let $\langle \cdot, \cdot \rangle$ be the standard inner product on $\bR^3$. The fiber of $FQ$ above $\bm{x} \in Q$ is defined as 
\begin{align*}
    F_{\bm{x}} Q \equiv \{u:\bR^3 \to T_{\bm{x}} Q \,:\,  u \text{ linear, orientation preserving,}\\
    \text{ and }
    \forall X,Y \in \bR^3, \langle X,Y\rangle  = g|_{\bm{x}}(u(X),u(Y)) \} \;.
\end{align*}
That is, $F_{\bm{x}} Q$ is the space of linear isometries between $(\bR^3,\langle \cdot,\cdot\rangle)$ and $(T_{\bm{x}} Q, g|_{\bm{x}})$. Equivalently, $F_{\bm{x}} Q$ is the space of orthonormal frames at $\bm{x}$ with the isomorphism stemming from the identification $u_i \equiv u(e_i)$ for the standard basis $e_i$ of $\bR^3$. Taking the disjoint union 
\[
    FQ \equiv \bigsqcup_{\bm{x}\in Q} F_{\bm{x}} Q \;,
\]
we have the bundle of orthonormal frames with natural projection $\pi_F:FQ \to Q$. 

There is an intrinsically defined $\bR^3$-valued $1$-form on $FQ$ called the \textit{tautological $1$-form} $\omega$ (also known as the soldering or fundamental $1$-form) \cite{Ivey16}. This form is defined pointwise for $(\bm{x},u)\in FQ$ by
\[
    \omega|_{(\bm{x},u)}(X) = u^{-1} \left( d\pi_F|_{(\bm{x},u)} X\right) \;.
\]
It has the useful property that it generates the dual orthonormal coframe to any orthonormal frame $u = (u_1, u_2, u_3)$ at a point $\bm{x}\in Q$. 
Indeed, define $\omega_u(X) \equiv \omega|_{(\bm{x},u)}(\tilde{X})$ for any $\tilde{X} \in T_{(\bm{x},u)} FQ$ such that $d\pi_F|_{(x,u)} \tilde{X} = X$. Then
\[
    \omega_u^i(u_j) = \langle e_i, \omega_u(u_j) \rangle = \langle e_i, \omega|_{(x,u)} (\tilde{u}_j) \rangle = \langle e_i, u^{-1} u_j \rangle = \langle e_i,e_j \rangle = \delta_{ij} \;.
\]
Hence $\omega_u = (\omega_u^1,\omega_u^2, \omega_u^3)$ is the dual coframe to $u$. 

More importantly, $\omega$ satisfies Cartan's structure equations 
\begin{equation} \label{eq:CartanStructure_1}
\begin{aligned}
    d\omega^i &= - \sum_{j=1}^3 \omega^i_j\w \omega^j\;, \\
    d\omega^i_j &= - \sum_{k=1}^3 \omega^i_k \w \omega^k_j + K^i_j\;,
\end{aligned} 
\end{equation}
where $\omega^i_j$ are the \textit{Levi-Civita connection} 1-forms on $FQ$, and $K^i_j$ are the curvature 2-forms on $FQ$. For a flat metric, the curvature 2-forms vanish: $K^i_j = 0$. Moreover, due to the orthonormality of the frame bundle and metric compatibility, $\omega^i_j = -\omega^j_i$. Hence, there are three linearly independent connection $1$-forms, $\omega^3_2, \omega^3_1, \omega^1_2$. Finally, it can be shown that $\omega^i, \omega^i_j$ are independent so that they form a coframing of $FQ$. 

The coframe $\{\omega^i,\omega^i_j\}$ of $FQ$ can be used to produce a coframe on $SQ$ that retains the nice properties of Cartan's structure equations \Eq{CartanStructure_1}. Define 
\[
    \Psi: FQ \to  SQ,\qquad \Psi(x,u) = (x, u(e_3)) \;.
\]
Note that $\Psi$ is well-defined because $g(u(e_3),u(e_3)) = \langle e_3, e_3 \rangle = 1$, implying that $u(e_3) \in S_{\bm{x}} Q$. The map $\Psi$ gives $FQ$ the structure of a fiber bundle with base space $SQ$ and fibers 
\[
    \Psi^{-1}(\bm{x},\bm{b}) = \{u\in F_{\bm{x}}Q \,:\, u(e_3) = \bm{b} \} \cong {\rm SO}(2) \;.
\]
This is diffeomorphic to ${\rm SO}(2)$ because orthonormal frames with $u(e_3) = \bm{b}$ have a freedom given by an oriented rotation of the 2-plane orthogonal to $\bm{b}$. 
A choice of a smooth orthonormal oriented frame $\cG(\bm{x},\bm{b}) = (\cG_1(\bm{x},\bm{b}), \cG_2(\bm{x},\bm{b}), \bm{b})$ for all $(\bm{x},\bm{b}) \in U\subset SQ$ is a choice of gauge on the set $U$. The choice of gauge can be considered a local section $\cG:U\subset SQ\to FQ$ of $\Psi$. 

For any local section $\cG:U\subset SQ\to FQ$ of $\Psi$, define
\begin{equation}\label{eq:omegabarDef}
    \bar{\omega}^i \equiv \cG^* \omega^i,\qquad \bar{\omega}^i_j = \cG^*\omega^i_j \;.
\end{equation}
The connection 1-forms for a chosen gauge are explicitly given by 
\[ 
    \bar{\omega}^i_j(X) = g\left( \bar{\nabla}_X\cG_j,\cG_i\right) \;,
\]
where, if $(E_1,E_2,E_3)$ is any local orthonormal frame on $Q$, $\nabla$ is the Levi-Civita connection, and $ \cG_j=\sum_{a=1}^3 R^a_jE_a$, then \[ \bar{\nabla}_X\cG_j \equiv \sum_{a=1}^3 X(R^a_j)E_a + R^a_j\nabla_{d\pi_S(X)}E_a. \] 
This definition is independent of the choice of the frame $(E_1,E_2,E_3)$.

The following lemma shows that a choice of gauge gives the desired coframing of $SQ$.

\begin{lemma}\label{lem:fExists}
    For any local section $\cG:U\subset SQ\to FQ$ of $\Psi$, define $\bar{\omega}^i,\bar{\omega}^i_j$ as in \Eq{omegabarDef}.
    The following properties hold:
    \begin{enumerate}
        \item The 1-forms $\bar{\omega}^1, \bar{\omega}^2,\bar{\omega}^3$ satisfy \Eq{omegaProp}: $\bar\omega^3 = \bm{b}^\flat$ and $
        \bar\omega^1 \w \bar\omega^2 = \iota_{\bm{b}} \Omega.$
        \item The 1-forms $\bar{\omega}^i, \bar{\omega}^i_j$ have structure equations
        \begin{equation}\label{eq:CartanCurvature}
        \begin{aligned}
            d\bar\omega^i &= -\sum_{j=1}^3 \bar\omega^i_j\w \bar\omega^j \;, \\
            d\bar\omega^i_j &= - \sum_{k=1}^3 \bar\omega^i_k\w \bar\omega^k_j + \bar{K}^i_j \;,
        \end{aligned}
        \end{equation}
        where $\bar\omega^i_j = -\bar\omega^j_i$ and $\bar{K}^i_j \equiv \cG^* K^i_j$. In particular, $\bar{K}^i_j = 0$ if $g$ is flat. 
        \item The 1-forms $\{\bar{\omega}^1,\bar{\omega}^2,\bar{\omega}^3,\bar{\omega}^3_1,\bar{\omega}^3_2\}$ are a coframe of $U$. Moreover, if $g$ is flat on $\pi_S(U)$ then every $(\bm{x},\bm{b}) \in U$ has a neighborhood $\tilde{U}\subset U$ and a choice of gauge $\cG:\tilde{U} \to FQ$ for which 
        \begin{equation}
            \bar{\omega}^1_2 = w_1\, \bar\omega^3_1 + w_2\, \bar\omega^3_2,
        \end{equation}
        for some functions $w_1,w_2:\tilde{U} \to \bR$.
    \end{enumerate}
\end{lemma}
\begin{proof}
    Property 1 follows from the definition of $\bar{\omega}^i$, the fact that $\pi_F\circ \cG = \pi_S$, that $\pi_F^*\Omega = \omega^1\w\omega^2\w\omega^3$, and that $\cG_3(\bm{x},\bm{b}) = \bm{b}$. 
    
    Property 2 holds by the properties of the pullback, in particular, $\cG^* \circ d = d\circ \cG^*$, and Cartan's structure equations for $\omega^i,\omega^i_j$ on $FQ$.

    The common kernel of $\bar{\omega}^1,\bar{\omega}^2,\bar{\omega}^3$ is $\ker d\pi_S$, since $\bar{\omega}_{(\bm{x},\bm{b})} = \cG(\bm{x},\bm{b})^{-1}\, d\pi_S$. On the two-dimensional $\ker d\pi_S$, the 1-forms $\bar{\omega}^3_1,\bar{\omega}^3_2$ describe the variation of $\bm{b}$ in the directions orthogonal to $\bm{b}$. Therefore the five forms are independent. Since $\dim(SQ)=5$, these forms give a local coframe, and the remaining connection form $\bar\omega^1_2$ can be written as a linear combination of all five. 
    
    If $g$ is flat near $\bm p$ on $Q$, then on a neighborhood $V$ of $\bm p$, there is a local parallel oriented orthonormal frame $E(\bm{x}):\bR^3 \to T_{\bm{x}}Q$, which locally identifies $SQ$ with $V\times S^2$. Choose a local map $R:W\subset S^2 \to \rm{SO}(3)$ with $R(\bm{b})e_3 = \bm{b}$, and define the gauge $\cG:\tilde{U} \equiv V\times W \to FQ$ by
    \[  \cG(\bm{x},\bm{b})=(E(\bm{x})R(\bm{b}) e_1, E(\bm{x}) R(\bm{b}) e_2, E(\bm{x})R(\bm{b}) e_3). \]
    Since $E$ is parallel, the definition of the Levi-Civita connection forms gives 
    $ \bar{\omega}^i_j(X) = \left\langle R(\bm{b})^{-1}dR(X) \, e_j , e_i \right\rangle. $
    Since $R$ depends only on $\bm{b}$, the forms $\bar{\omega}^i_j$ vanish on vectors tangent to $V$. In particular, $\bar{\omega}^1_2$ is vertical with respect to the projection $\tilde{U}=V\times W\to V$. Differentiating $R(\bm{b})e_3=\bm{b}$ gives 
    \[ d\bm{b} = dR\,e_3 = R(\bm{b})(R(\bm{b})^{-1}dR)e_3 = -\bar{\omega}^3_1R(\bm{b})e_1 -\bar{\omega}^3_2R(\bm{b})e_2. \]
    Hence $\bar{\omega}^3_1,\bar{\omega}^3_2$ form a coframe on the vertical tangent spaces. Since $\bar{\omega}^1_2$ is vertical, there exist smooth functions $w_1,w_2$ on $\widetilde{U}$ such that $\bar{\omega}^1_2 = w_1\bar{\omega}^3_1+w_2\bar{\omega}^3_2$.
\end{proof}

\section{Pfaffian systems from PDE and EDS\label{app:PDE_to_Pfaffian}}
Here we describe two important ways in which Pfaffian systems arise in practice. In the first case we are given a system of partial differential equations and construct a Pfaffian system whose associated PDE \Eq{general_pfaffian_pde} is equivalent to the original system. In the second case we are given a so-called exterior differential system and apply prolongation to find an equivalent Pfaffian system on a higher-dimensional space.

Consider a general first-order system of $N$ partial differential equations for an unknown field $\depvar:\bR^n\rightarrow\bR^d$, 
\begin{align}\label{eq:NPDEs}
    \cF^a(\bm{x},\depvar(\bm{x}),D\depvar(\bm{x})) = 0,\quad a = 1,\dots, N,
\end{align}
where each $\cF^a$ is a scalar function.
Introduce the jet bundle $\widetilde{M} = \bR^n\times\bR^d\times\bR^{d\times n}\ni (\bm{x},\depvar,\bm{p})$, whose points represent values of $\depvar$ and its first partial derivatives. Solutions to the PDEs \Eq{NPDEs}, define a submanifold $M\subset\widetilde{M}$. There is a natural Pfaffian system defined on $M$ given by pulling back the contact $1$-forms $\widetilde{\alpha}^a = d\depvar^a - p^a_idx^i$ from $\widetilde{M}$ to $M$. The PDE \Eq{general_pfaffian_pde} associated with this Pfaffian system is equivalent to the original PDE system \Eq{NPDEs}. We can represent this Pfaffian system in the form described in \Sec{Pfaffian_Defined} as follows. Suppose in a given basis $\bR^d\times\bR^{d\times n}\ni (\depvar,\bm{p})$ may be written $\bR^\ell\times \bR^{d+d\times n-\ell}\ni (\bm{y},\bm{z})$ and assume \Eq{NPDEs} determines $\bm{z}$ as a function of $\bm{x}$ and $\bm{y}$. Then $M$ has coordinates $(\bm{x},\bm{y})\in\bR^n\times\bR^\ell$ and pulling back the $\widetilde{\alpha}^a$ to $M$ expresses the Pfaffian system in the form $\alpha^a = (\alpha^a_{\bm{x}})^Td\bm{x} + (\alpha^a_{\bm{y}})^Td\bm{y}$, $a = 1,\dots, r$, with $r = d$. A similar construction leads to a Pfaffian system formulation for a system of partial differential equations of any order.

Now consider an exterior differential system (EDS) on a manifold $M$, i.e. a finitely-generated differential ideal $\mathcal{I}$ in $M$'s exterior algebra. Denote the generating differential forms $\lambda_k$, $k=1,\dots, g$. An $n$-dimensional integral manifold is a submanifold of $M$ on which the pullback of every $\lambda_k$ vanishes. An integral $n$-element (or just integral element when $n$ is understood) at $m\in M$ is an $n$-dimensional subspace of $T_mM$ on which the pullback of every $\lambda_k(m)$ vanishes. The collection of all integral elements at $m$ is denoted $\cV_m$. The space of all integral elements over $M$, $\cV = \cup_{m\in M}\cV_m$, is generally a closed subset of the Grassman bundle $\text{Gr}(TM)$ over $M$. Suppose that $\cV_0 \subset\cV$ is a smooth submanifold of $\text{Gr}(TM)$. Pulling back the contact $1$-forms on $\text{Gr}(TM)$ to $\cV_0$ defines a Pfaffian system on $\cV_0$ called the $\cV_0$-prolongation of $\mathcal{I}$. The integral $n$-manifolds of the first prolongation are always in one-to-one correspondence with integral $n$-manifolds of the original EDS whose integral elements lie in $\cV_0$. 

\end{document}